\documentclass[letterpaper,11pt]{article}
\usepackage[sc]{mathpazo}
\usepackage{microtype}
\usepackage[utf8]{inputenc}
\usepackage{graphicx,fullpage,paralist}
\usepackage{amsmath, amssymb, amsthm}
\usepackage{comment,hyperref}
\usepackage{thmtools}
\usepackage{tikz}
\usepackage{pgfplots}
\usepackage{varwidth}
\pgfplotsset{compat=1.10}
\usepgfplotslibrary{fillbetween}
\usetikzlibrary{backgrounds}
\usetikzlibrary{patterns}
\usepackage{enumitem}
\usepackage{geometry}
\usepackage{array}
\usepackage{multirow}
\usepackage{stmaryrd}
\usepackage[font=small]{caption}
\usetikzlibrary{shapes.geometric}
\usetikzlibrary{arrows}
\usetikzlibrary{snakes,automata,patterns}

\newcommand{\calR}{\mathcal{R}}

\newcommand{\RR}{\mathbb{R}}
\newcommand{\NN}{\mathbb{N}}

\let\epsilon\varepsilon

\DeclareMathOperator{\dist}{\mathrm{dist}}

\newcommand{\R}{\mathbb{R}}
\newcommand{\N}{\mathbb{N}}

\newcommand{\Z}{\mathbb{Z}}
\newcommand{\E}{\mathbb{E}}
\newcommand{\TV}{\mathsf{TV}}
\renewcommand{\P}{\mathbb{P}}

\newcommand{\M}{\mathcal{M}}

\newcommand{\mix}{\mathrm{mix}}

\newtheorem{theorem}{Theorem}
\newtheorem{lemma}{Lemma}
\newtheorem{corollary}{Corollary}

\newtheorem{definition}{Definition}

\newtheorem{proposition}{Proposition}

\usepackage[textsize=tiny,textwidth=2cm]{todonotes}

\usepackage{multirow}
\usepackage[noend]{algpseudocode}
\usepackage{algorithm,algorithmicx}
\usepackage{xcolor}

\newlength{\algofontsize}
\hypersetup{
	colorlinks,
	linkcolor={red!50!black},
	citecolor={blue!50!black},
	urlcolor={blue!80!black}
}

\begin{document}
	\algrenewcommand\algorithmicrequire{\textbf{Input:}}
	\algrenewcommand\algorithmicensure{\textbf{Output:}}
	
	\title{The Convergence Rates of Blockchain Mining Games:\\ A Markovian Approach
    \vspace{.3cm}
	 }
	\author{Alejandro Jofr\'e
	\thanks{Department of Mathematical Engineering, Universidad de Chile. {\tt ajofre@dim.uchile.cl}}
	\and Angel Pardo			
	\thanks{Department of Mathematical Engineering, Universidad de Chile. {\tt aapardo@dim.uchile.cl}}
	\and David Salas
	\thanks{Institute of Engineering Sciences, Universidad de O'Higgins. {\tt david.salas@uoh.cl}}
	\and Victor Verdugo
	\thanks{Institute of Engineering Sciences, Universidad de O'Higgins. {\tt victor.verdugo@uoh.cl}}
	\and Jos\'e Verschae
	\thanks{Institute of Mathematical and Computational Engineering, Pontificia Universidad Cat\'olica. {\tt jverschae@uc.cl}}
	}
\date{\vspace{-1em}}

\maketitle
\begin{abstract}
Understanding the strategic behavior of miners in a blockchain is of great importance for its proper operation. A common model for mining games considers an infinite time horizon, with players optimizing asymptotic average objectives. Implicitly, this assumes that the asymptotic behaviors are realized at human-scale times, otherwise invalidating current models. 
We study the mining game utilizing Markov Decision Processes. Our approach allows us to describe the asymptotic behavior of the game in terms of the stationary distribution of the induced Markov chain. We focus on a model with two players under immediate release, assuming two different objectives: the (asymptotic) average reward per turn and the (asymptotic) percentage of obtained blocks.

Using tools from Markov chain analysis, we show the existence of a strategy achieving slow mixing times, exponential in the policy parameters. This result emphasizes the imperative need to understand convergence rates in mining games, validating the standard models. Towards this end, we provide  upper bounds for the mixing time of certain meaningful classes of strategies. This result yields criteria for establishing that long-term averaged functions are coherent as payoff functions. Moreover, by studying hitting times, we provide a criterion to validate the common simplification of considering finite states models.
For both considered objectives functions, we provide explicit formulae  depending on the stationary distribution of the underlying Markov chain. In particular, this shows that both mentioned objectives are not equivalent. Finally, we perform a market share case study in a particular regime of the game. More precisely, we show that an strategic player with a sufficiently large processing power can impose negative revenue on honest players. 
\end{abstract}

\thispagestyle{empty}
\newpage

\section{Introduction}

As the use of cryptocurrencies continue to grow, understanding the behavior of its miners becomes paramount. A typical approach to analyze the Bitcoin mining process considers a game-theoretical setting where selfish miners aim to optimize an average objective. This has been considered in the influential work by Eyal and Sirer~\cite{eyal2014majority}, Kiayias et al.~\cite{Koutsoupias2016} and several subsequent works~\cite{eyal2015miner,sapirshtein2016optimal,nayak2016stubborn,goren2019mind,koutsoupias2019blockchain,arenas2020cryptocurrency,leonardos2020oceanic}. Most models study this setting as a sequential game with infinite time horizon, where each player optimizes an asymptotic payoff. An implicit assumption on this definition, is that these asymptotic objectives are realized at human time scales. Our long term purpose is to understand the convergence rates to steady states of the mining game, hence validating mathematically these definitions if possible, and more importantly, understanding their limitations. 

At an intuitive level, the main argument for the pertinence of such payoff functions and the infinite horizon assumption is a natural Markovian structure of the sequential game: once blocks are validated by all miners, they become part of the official blockchain and they are no longer in competition. However, the underlying Markovian process that miners face has, to the best of our knowledge, never been carefully formalized.
The ergodic properties of the underlying Markov process allow to study the payoff functions asymptotically. However, Markovian properties alone do not guarantee fast convergence. Thus, a natural question arises: Are asymptotic payoffs actually observable in a time scale that is reasonable for the miners? That is, can we see the long term behavior at human time scale? Once this question is on the table, the assumption of infinite time horizon, together with other ones like constant rewards and no simultaneous mining, become questionable. Indeed, taking as example the Bitcoin protocol (well surveyed in \cite{survey2016bitcoin}), the rewards of mining a block (which are new Bitcoins that are created) halves every four years, starting at 50 BTC on 2009. By July 2021, the reward of a block is 6.25 BTC and it is projected that new Bitcoins rewards will end the year 2140, when the reward of a block becomes $10^{-8}$ BTC (see \cite{survey2016bitcoin}). This discount rate tells us that the assumption of constant rewards is valid only if long run means less than 4 years.  

Regarding the Markovian structure of the sequential game, it is usual in the literature~\cite{eyal2014majority,Koutsoupias2016,eyal2015miner,sapirshtein2016optimal,nayak2016stubborn,goren2019mind,koutsoupias2019blockchain,arenas2020cryptocurrency,leonardos2020oceanic} to consider finitely many possible states in the chain representation of it. These states are given by the blocks that have not been recognized by all miners. The assumption rests on the original blockchain white paper \cite{nakamoto2019bitcoin}, where it is stated that blocks in dispute tend to be resolved naturally by the consensus of honest miners, choosing one branch as the official an leaving the blocks in the other ones as orphans. Thus, truncated models are applied, where it is assumed that no strategic miner will persevere beyond certain threshold $d$: If another branch reaches the length $d+1$, he or she will capitulate, abandon his or her branch as orphan. However, this threshold is not part of the blockchain protocol, and miners could potentially deviate from this assumption. Thus, another natural question is when truncated models are consistent, in the sense that they predict correctly the evolution of the sequential game.

In order to study these questions, we first formalize the mining game by the lens of a Markov Decision Process (MDP). This is not an easy task, mainly due to the fact that strategic miners can hide information. Thus, as a starting point for this line of work, we focus our attention on mining games under perfect information, or, as it was called in \cite{Koutsoupias2016}, \emph{immediate release}. Once this is done, we will use the explicit Markov chain structure to answer the questions described above in using the theory of mixing and hitting times. 

\subsection{Our Contribution}

We start by formalizing the mining game through the lens of MDP's, allowing us to analyze the mining game through the theory of Markov chains.
In particular, we can describe some of the relevant payoff functions considered in the literature by understanding the stationary distribution of the underlying chain. 
The formal definition of the game and the payoffs can be found in Section~\ref{sec:Pre}, while the Markovian approach is formalized in Section~\ref{sec:MarkovModel}. 

We consider the case of two players. In this case, both players can mine either the same block, or mine their own branch of the blockchain. After each mining race, they can either decide to keep mining their branch or capitulates to some block of the other branch. Hence, the state of the game are given by two parameters $(\ell_1,\ell_2)$, which represent the length of the two branches from the last common ancestor. The states of the Markov chains can be represented within the integer lattice, which allows us to find explicit expressions of the stationary distributions by solving a combinatorial problem corresponding to counting interior paths. 
We exploit this representation to provide an example showing that certain strategies can have large mixing times. This highlights the importance of understanding mixing times, and hence the validity of the models, for different strategies.
We are also able to give upper bounds on the mixing time, which helps us evaluate convergence rates to the stationary regime. We evaluate this measure for the case in which one player plays honestly, while the other capitulates if the length difference of the two branches is larger than a given parameter $\overline{g}$. Our  upper bound of the mixing time depends on the processing power of the strategic player and $\overline{g}$. The proof is obtained by coupling techniques. 
Our analysis allows us to validate the common assumption of asymptotic payoff functions for these strategies. 

Through our representation, we formalize a notion of \emph{safety}, which captures the idea of when selfish mining can be based on simplified models, at no risk of violating the blockchain protocol, that is, parallel validated blocks. 
This is performed by studying the hitting time of the critical state of having two parallel branches of length $d$. We provide sufficient conditions under which a strategy can be safely played, and we show the existence of unsafe strategies. The study is carried out by bounding the number of interior paths from the initial states to the critical points in the lattice representation.
The mixing and hitting time study can be found in Section \ref{sec:Times}.

Finally, in Section \ref{sec:Revenues} we perform a market share study over a particular family of strategies, based on our Markovian model and the convergence analysis, in order to understand the impact on the profit of the miners according and how it behaves depending on the miner payoff function. In this setting, we study the two most common payoff functions, which are the asymptotic ratio of owned blocks, and the asymptotic revenue per time unit. We observe that these two functions are far from being equivalent. Furthermore, we observe that in our regimes, a sufficiently powerful strategic miner can apply a strategy that makes other miners to have negative profits, and therefore pushing them out of the mining game.

\subsection{Related Literature}
Since the introduction of the Bitcoin blockchain protocol in 2008 by Nakamoto \cite{nakamoto2019bitcoin}, there has been an intensive study around the development and understanding of the blockchain protocol and its variants. We refer to Abadi and Brunnermeier \cite{abadi2018blockchain} for an economic analysis of the blockchain protocol, and to \cite{survey2016bitcoin} for a technical review.
Selfish mining was originally studied by Eyal and Sirer \cite{eyal2014majority}, and they show the existence of a non-honest mining strategy that is profitable if the hash power of the non-honest players is at least 33.3\%.
Particularly relevant to our setting is the work by Kiayias et al. \cite{Koutsoupias2016}, that model strategic mining through the lens of a stochastic game under complete information.

Recently, Marmolejo-Coss\'io et al. \cite{marmolejo2019competing} extend the analysis of Eyal and Sirer to the case of multiple, not necessarily colluded, selfish miners, showing that they have incentives to deviate in block from honest mining. 
Arnosti and Weinberg \cite{arnosti2019bitcoin} study the Bitcoin mining  through the lens of investment costs vs market concentration, finding evidence of an oligopolistic behavior.
The study of strategic and selfish mining is a very active area; see e.g.~\cite{eyal2015miner,sapirshtein2016optimal,nayak2016stubborn,goren2019mind,koutsoupias2019blockchain,arenas2020cryptocurrency,leonardos2020oceanic}.
Recently, there have been efforts in studying variations of the blockchain protocol and to analyze the impact in terms of mining incentives \cite{birmpas2020fairness,pass2017fruitchains}. 


\section{Preliminaries: The Mining Game \label{sec:Pre}}
In this section we describe the model of the sequential game that is played by the miners in the Blockchain protocol. This model was already implicit in~\cite{eyal2014majority} and it is fully described in~\cite{Koutsoupias2016}. 
We consider a sequential game of $N$ players called \textit{miners}, each miner having a computing capacity. 
We assume that the computing capacity remains constant during the game and that players stay in the mining game if it is profitable for them. That is, we do not consider \emph{in-and-out} strategies, where some players might strategically stop mining for some turns to temporarily reduce the total computing power of the game (this strategic behavior has been studied in \cite{GrunspanPerez-Marco2018}). 

At each turn $n$, the blockchain is given by a public  rooted directed tree of blocks $T$. Each block $B$ is labeled with $i\in \{1,\ldots,N\}$, which is the player that wrote it. Miners are rewarded for writing blocks. At each turn, each player $i$ has a private tree $T_i$ that might differ from $T$ and a current mining block $B_i$ in $T_i$.  Players are trying to solve a crypto-puzzle associated to their mining blocks, in order to find a key to attach (write) a new block.  

\begin{definition} At each turn, we say that a player $i\in \{1,\ldots,N\}$ has won the mining race if he or she is the first one in solving the crypto-puzzle associated to his or her current mining block $B_i$.
	We suppose that the mining race has only one winner almost surely, and we set $p_i$ to be the probability that player $i\in \{1,\ldots,N\}$ is the one winning the mining race.  
\end{definition}

In practice, there is a delay between the finding of a new block and the communication to all miners of this new block.
If a second miner finds a block in this time window, then this second block is considered as simultaneous. 
Here, however, we omit the possibility of multiple players finding simultaneous blocks. 
The probabilities $p_1,p_2,\ldots,p_N$ follow by modeling the time needed to solve the crypto-puzzles as exponential random variables, depending on the fixed computational capacity of each player and the difficulty of the puzzle.
The sequential game is then played as follows:
\begin{enumerate}[label={\texttt{\color{blue}Step \arabic*.}},ref= \texttt{{\color{blue}Step \arabic*}},align=left, leftmargin=*]
	\setcounter{enumi}{-1}
	\item\label{Step0}  At the beginning of the game ($n=0$), there is a public tree $T$ of only one block.
	Each player starts mining this block.
	\item \label{Step1}  A mining race is played.
	\item\label{Step2} The player $i\in \{1,\ldots,N\}$ that has found a block adds it to his or her private tree, $T_i$, and updates the current mining block to the new block found. Then, a revealing phase starts:
	\begin{enumerate}[label*=\texttt{{\color{blue}\arabic*:}}, ref = \texttt{{\color{blue}\ref{Step2}.\arabic*}},align=left, leftmargin=*]
		\item\label{RevealPhase1} Each player $i\in \{1,\ldots,N\}$ reveals part of its private tree to add some elements in $T_i\setminus T$ to the public tree. 
		\item \label{RevealPhase2} Every player updates the public tree. If the public tree is the same as before, the reveal phase ends. If not, we go back to \ref{RevealPhase1}.
	\end{enumerate}
	\item\label{Step3} Each player $i\in \{1,\ldots,N\}$ decides a (possibly new) block to mine, and updates its current mining block $B_i$ to the selected block. If the depth of $T$ is less than a maximum length $D_{\max}$, the game goes back to \ref{Step1}. Otherwise, the game ends.
\end{enumerate}
The mining game is said to be of immediate release if every player $i\in \{1,\ldots,N\}$ decides to reveal his private tree at the revealing phase, that is, at the beginning of \ref{Step3}, $T_i=T$ for every player $i\in \{1,\ldots,N\}$. In the following we define a particular type of ``honest" strategy, corresponding to a player immediately releasing any mined block and selecting one of the deepest
nodes.
This type of strategy was studied by Kiayias et al.~\cite{Koutsoupias2016}.
\begin{definition} We say that a player $i\in \{1,\ldots,N\}$ is playing the Frontier strategy if, at \ref{Step3}, he or she selects a leaf of the deepest level of the public tree $T$. If one of the leaves of the deepest level belongs to $i$, then  $i$ selects that block.
\end{definition}


\subsection{Payoff Functions \label{subsec:PayoffFunctions}}

In order to define the objective functions, we need to understand how blocks become part of the official blockchain. The blockchain protocol states that the official transaction history must be the longest path in the public tree $T$, which is the chain with most proof-of-work (see \cite{nakamoto2019bitcoin,survey2016bitcoin}). However, when forks appear (due to strategic mining or simultaneous mining), there is an ambiguity in the definition of the longest chain. Thus, users of the cryptocurrency consider a block valid once this block has a certain amount of children in the tree. 
This delay on validation of transactions contained in a block is the key element to prevent double spend-attacks. We formalize this idea with the following definition.

\begin{definition}\label{def:Validation} A block $B\in T$ is validated if all miners only select blocks that are descendants of $B$. A mining game has maximum depth $d\in \N\cup\{\infty\}$ if once a block $B$ has a path of $d$ descendants, then it becomes validated.
\end{definition}
In practice, there is not an official depth $d$ in the cryptocurrency protocols. However, assuming that miners are not seeking to perform double-spend attacks \cite{rosenfeld2014analysis}, an implicit maximum depth $d$ is considered, obtained consensually to prevent ambiguity in the transactions' history. For example, in the Bitcoin protocol, blocks are paid once a path of length $d=100$ follows them, and it is a reasonable maximum depth to consider. This simplifying hypothesis has been used before in \cite{Koutsoupias2016}.
In what follows we consider the following parameters in the game:
\begin{enumerate}
	\item $D_{\max}$ is considered as $\infty$. Each player $i$ has a marginal cost $c_i\geq 0$ of mining per time unit. 
	\item Each block has a reward $r>0$, which is collected by the player that owns it once the block is validated. The marginal costs $c_1,c_2,\ldots,c_N$ and the rewards are normalized so $r=1$.
\end{enumerate}

For every positive integer $n$ and player $i$, we denote by $r_{i,n}$ the number of blocks won by player $i$ that are validated at the end of turn $n$. We set $\dist_n$ to be the number of blocks that are validated at the end of turn $n$, and $\tau_n$ to be the time length of turn $n$. We model the sequence $(\tau_n)_{n\in\N}$ as independent identically (exponentially) distributed variables. 
This structure yields a natural probability space $(\Omega_0,\mathcal{F}_0,\P_0)$ where
$\Omega_0 = (\{1,\ldots,N\}\times \R_+)^{\N}$, a sequence $\omega = (\omega_n,t_n)_{n\in\N}\in \Omega_0$ represents the outcomes of a sequence of mining races together with their time lengths, and
$\mathcal{F}_0$ is the product $\sigma$-algebra $\bigotimes_{n\in\N} \mathcal{P}(\{1,\ldots,N\}\times\mathcal{B}(\R_+))$, where $\mathcal{B}(\R_+)$ is the Borel $\sigma$-algebra.
Then, $\P_0$ is the unique probability measure satisfying that
$\P_0(\omega_n = i) = p_i$ for every $i\in \{1,\ldots,N\}$ and every positive integer $n$, where $\{\omega_n = i\}$ means that the player $i$ won the mining race at turn $n$.

With this probability space, the law of large numbers ensures that
$\sum_{k=0}^{n-1}\tau_k/n\rightarrow   \tau_b$ almost surely, where $\tau_b=\E_0(\tau_0)$.
Thus, $\tau_b$ is the averaged time needed to find a new block, which we set to be the averaged time length of a turn. The blockchain protocol adjust the difficulty of mining a new block periodically (see e.g. \cite{survey2016bitcoin}) in such a way that a target average constant time per block is obtained. We can model this property by assuming that
\begin{equation}\label{eq:adjustment}
	\lim_{n\to\infty}\E_0\left( \frac{\sum_{k=0}^{n-1}\tau_k}{1+\sum_{k=0}^{n-1}\dist_k }\right) = \overline{\tau},
\end{equation}
where $\overline{\tau}$ is a constant time called \emph{target} (in Bitcoin, $\overline{\tau}$ is set at 10 minutes). Strategic behavior modify the distributions of $(\tau_n)_{n\in\N}$ and the average time $\tau_b$ of mining a block, by influencing the difficulty adjustment (see \cite{GrunspanPerez-Marco2018}). However, if strategic players maintain their strategies and do not perform in-and-out attacks, it is reasonable to assume that the sequence $(\tau_n)_{n\in\N}$ is identically distributed and that $\tau_b$ remains constant.

Now, we are ready to present the possible objective functions for miners in the game. In this work, we consider two possible objectives that players seek to maximize, both in the long term: The expected revenue per turn and the expected ratio of validated blocks.
More specifically, the three criteria are the following:
\begin{enumerate}[label=(\roman*)]
	\item A player $i$ maximizes the asymptotic revenue per turn 
	$R_i = \displaystyle\lim_{n\to \infty} \frac{1}{n}\E\left(\sum_{k=0}^{n-1} r_{i,k} - c_i\tau_k  \right)$.
	\item A player $i$ maximizes the asymptotic ratio of validated blocks, that is, \[G_i = \lim_{n\to \infty} \E\left( \frac{\sum_{k=0}^{n-1} r_{i,k}}{1+\sum_{k=0}^{n-1}\dist_k}\right),\]
	with the additional constraint that $R_i\geq 0$. Here the cost can be neglected since \eqref{eq:adjustment} ensures that the average costs per validated block become asymptotically constant, equal to $c_i\overline{\tau}$.
\end{enumerate}
The first work considering selfish mining \cite{eyal2014majority} assumes that players try to maximize the ratio of owned blocks in the official chain, that is, each player $i$ is maximizing $G_i$ (i.e. their market share). This objective was also considered in other earlier influential work \cite{Koutsoupias2016}, and it has been generally accepted as one of the natural payoff functions in the mining game. As we will see, the choice of the objective function has a deep influence on the strategic behavior of the players. In a nutshell, maximizing $G_i$ might induce that players benefit by diminishing the total number of available blocks (i.e., reduce $1+\sum_{k=0}^{n-1}\dist_k$), and it is not necessarily related to maximizing the revenues $R_i$. 

The following proposition states that Frontier strategies are a Nash equilibrium when only two players are trying to maximize $R_i$ without difficulty adjustment. This result was already stated by \cite{GrunspanPerez-Marco2018} and here we recall it to show the impact of the payoff function. Other studies of profitability including the impact of difficulty adjustment can be found in \cite{profitability2020albrecher,profitability2020davidson}.
\begin{proposition}[{\cite[Theorem 4.4]{GrunspanPerez-Marco2018}}]
	\label{prop:Frontier-Nash}
	For $N=2$, when both players try to maximize their asymptotic expected revenue $R_1$ and $R_2$, Frontier strategy is a Nash equilibrium.
\end{proposition}

\subsection{Discussion of the Model}
In this work we aim to build the foundations for the study of fast convergence rates of long-term payoff functions, as well as for the validity of truncated models. In particular, we aim to determine how strategic mining can be performed under such regimes.
The main assumptions of the model we presented in this section are the following:
\begin{enumerate}[label={{\color{blue}(A.\arabic*)}},ref= {\color{blue}A.\arabic*},align=left, leftmargin=*]
	\item\label{A1:Computational} The computational capacity of each agent is fixed, and therefore so it is his or her probability of winning a mining race. $D_{\max}$ is very large, and can be assumed to be $\infty$.
	\item \label{A3:FixedReward} The reward of adding a block to the official blockchain is fixed.
	\item \label{A4:Depth} There exists a maximum depth $d$, after which a block with $d$ descendants becomes validated by all miners.
	\item \label{A5:Simulataneous} There is no simultaneous mining. All miners play with immediate release, that is, there is no strategic revealing.
\end{enumerate}

Assumptions \eqref{A1:Computational}-\eqref{A3:FixedReward} are related with time scales and allow us to validate the payoff functions presented in Section \ref{subsec:PayoffFunctions}. The first two are quite standard and appear in several works studying the mining game \cite{Koutsoupias2016, marmolejo2019competing}. 
Assumption \eqref{A3:FixedReward} has been used in several studies, but it is definitely more debatable. Again in the Bitcoin protocol, as in February 2021, mined blocks are rewarded with $r=2.56$ BTC. The classic blockchain protocol has a discount rate on the reward of mining a block: in Bitcoin, it is reduced to the half every $210000$ validated blocks, starting at 50 BTC. This reduction, which happens around every 4 years, will continue until the reward becomes $10^{-8}$ BTC, and it is expected to happen in year 2140 (see \cite{survey2016bitcoin}).

The reader can observe that, in order to safely consider assumptions \eqref{A1:Computational}-\eqref{A3:FixedReward}, it is necessary that the time horizon at which the assumptions are no longer valid must be large with respect to the scale of time of the game (minutes). However, this requirement is not sufficient, since these assumptions are in competition with the convergence rate of the asymptotic payoff functions. Indeed, if one aims to consider limit objective functions, which represent long-term goals, it is also necessary that those limits are perceived within the scale of time  where the assumptions are valid. This fast convergence requirement is assured, through the lens of a Markovian approach, by a fast mixing of the dynamics of the game, and it is studied in Subsection \ref{subsec:MixingTime}.

Assumption \eqref{A4:Depth} is used to truncate the feasible states of the dynamics to a finite description. As we already mentioned, this assumption is artificial in some sense, since it is not part of the blockchain protocols. However, it simplifies the study of the game, and even when it has been used before, there is no validation, to the best of our knowledge, as a coherent approximation of the real dynamic without this truncation. Based on the theory of Markov processes, we provide a criteria to safely use the truncated model: The expected time needed to visit the states in the boundary of the truncation must be exponentially large. We study this criterion in Subsection \ref{subsec:Hitting}.
Assumption \eqref{A5:Simulataneous}
is a simplifying condition, which we expect to overcome in subsequent work.


\section{A Markov Model for a Mining Game with Two Players}
\label{sec:MarkovModel}

In what follows, we focus our attention in a mining game under immediate release with two players. When only one player (or colluded pool of players) is mining strategically, the honest players can be reduced to only one that concentrates the computational power. This reduction has been done previously in \cite{eyal2014majority,Koutsoupias2016} to study if Frontier strategy is an equilibrium. In this same line, while immediate release is a simplified model, since strategic players hide information, as stated in \cite{Koutsoupias2016} and studied in \cite{eyal2014majority,marmolejo2019competing}, it is a starting point to develop the theory of convergence rates.

Following the description of the mining game of Section \ref{sec:Pre} under the immediate release assumption, every state of the 2-players game is given by the tuple $(T,B_1,B_2)$, that is, the public tree and the mining blocks of each player.  By Definition \ref{def:Validation}, the only relevant information of $T$ is given by the subtree rooted at the last validated block. Assuming rationallity of both players, this subtree has only two branches: The path mined by player $1$ of depth $\ell_1$, and the path mined by player $2$ of depth $\ell_2$.  Thus, $(T,B_1,B_2)$ can be compactly represented by a pair $(\ell_1,\ell_2)$ (see Fig. \ref{fig:Capitulation}).
If one of the branches has a length strictly larger than $d$, it would mean that the first block of this branch has been validated. 
According to the truncated model, the other player must recognize this block as valid, and so the root of the tree must be a descendant of this block. 
Therefore, all possible states of the stochastic process are given by the integer pairs $(\ell_1,\ell_2)$ such that $0\leq\ell_1,\ell_2\leq d$.
At each turn, both players make a decision in \ref{Step3}, concerning which blocks to mine in order to find their new blocks. 
These decisions depend on two factors: the state $(\ell_1,\ell_2)$ of the game at the beginning of the turn, and the result of the mining race. 
Assuming rationality of both players, there are possibilities: To continue the branch or to capitulate it. 

On the one hand, to continue the branch means to mine the deepest block. 
On the other hand, to capitulate the branch means to select a block on the branch of the other player to restart the mining process. 
See Figure \ref{fig:Capitulation} that represents a capitulation of player $1$.
When a player capitulates, the state of the game is reset to $(0,s_1)$ if player $1$ capitulates, or to state $(s_2,0)$ if player $2$ capitulates. 
It is natural that for a given state $(\ell_1,\ell_2)$ and as a result of the mining race, at most one of the players capitulates (the one who loses the mining race).
As Figure \ref{fig:Capitulation} shows, the value $s_i$ with $i\in \{1,2\}$, is the amount of blocks that player $i$ will try to surpass after capitulation. We define a \emph{round} of the game as a set of transitions starting from one of the initial states $(0,s_1)$ and $(s_2,0)$ and a capitulation of one of the players. Players collect the rewards at the end of the rounds, and only one player wins the round (has positive revenue), which is the one that does not capitulate.

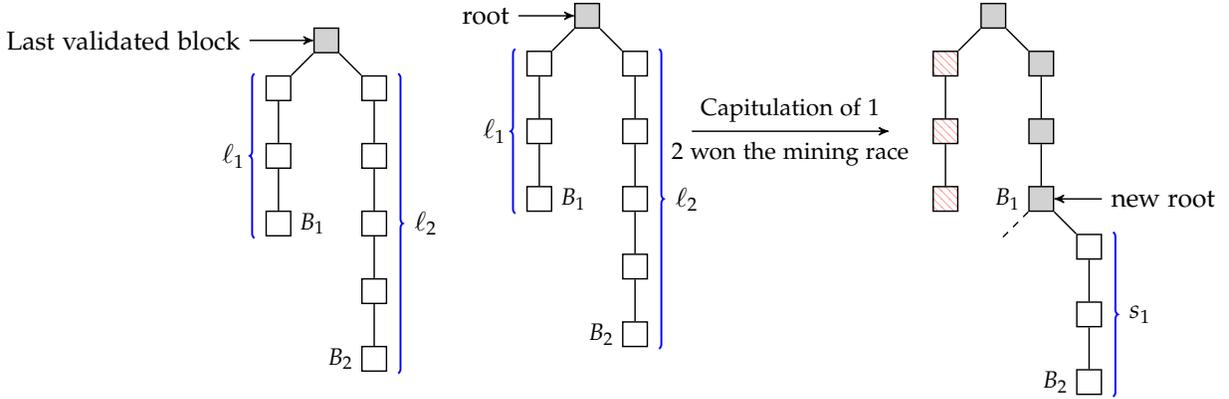
\begin{figure}[h]
	\begin{minipage}{.35\textwidth}
	\centering
	\scalebox{.9}{
	\begin{tikzpicture}[->,>=stealth',shorten >=0pt,auto,node distance=1cm,
		semithick,square/.style={regular polygon,regular polygon sides=4}]
		
		\tikzstyle{every state}=[square, text=black,font = \bfseries, fill opacity = 0.35, text opacity = 1,minimum size=0.5cm]
		\node[state] (Root) [fill = gray] {};
		\node[state] (B11) [below left of = Root] {};
		\node[state] (B12) [below of = B11] {};
		\node[state] (B13) [below of = B12] {};
		\node[] (MB1) [right of = B13,node distance=0.5cm] {\small $B_1$};
		\node[] (textRoot) [left of = Root, node distance=3cm] {Last validated block};
		\path (textRoot) edge (Root);
		
		\node[state] (B21) [below right of = Root] {};
		\node[state] (B22) [below of = B21] {};
		\node[state] (B23) [below of = B22] {};
		\node[state] (B24) [below of = B23] {};
		\node[state] (B25) [below of = B24] {};
		\node[] (MB2) [left of = B25,node distance=0.5cm] {\small $B_2$};
		
		\path[-]  (Root) edge (B11)
		(B11) edge (B12)
		(B12) edge (B13)
		;
		\path[-]  (Root) edge (B21)
		(B21) edge (B22)
		(B22) edge (B23)
		(B23) edge (B24)
		(B24) edge (B25)
		;	
		\node[] (L-B11) [above left of =B11, node distance = 0.5cm] {};	
		\node[] (L-B13) [below left of =B13,node distance = 0.5cm] {};	
		\node[] (R-B21) [above right of =B21, node distance = 0.5cm] {};	
		\node[] (R-B25) [below right of =B25,node distance = 0.5cm] {};		 	  
		\draw [-,thick, blue,decorate,decoration={brace,amplitude=2pt,mirror},xshift=-0pt,yshift=0](L-B11) -- (L-B13) node[left,black,midway,xshift=0] {$\ell_1$};
		\draw [-,thick, blue,decorate,decoration={brace,amplitude=2pt,mirror},xshift=0pt,yshift=0](R-B25) -- (R-B21) node[right,black,midway,xshift=0.1cm] {$\ell_2$};
	\end{tikzpicture}}
	\end{minipage}
	\begin{minipage}{.65\textwidth}
		\centering
		\scalebox{.9}{
	\begin{tikzpicture}[->,>=stealth',shorten >=0pt,auto,node distance=1cm,
		semithick,square/.style={regular polygon,regular polygon sides=4}]
		
		\tikzstyle{every state}=[square, text=black,font = \bfseries, fill opacity = 0.35, text opacity = 1,minimum size=0.5cm]
		\node[state] (Root) [fill = gray] {};
		\node[state] (B11) [below left of = Root] {};
		\node[state] (B12) [below of = B11] {};
		\node[state] (B13) [below of = B12] {};
		
		\node[] (textRoot) [left of = Root, node distance=1.5cm] {root};
		\path (textRoot) edge (Root);
		
		\node[state] (B21) [below right of = Root] {};
		\node[state] (B22) [below of = B21] {};
		\node[state] (B23) [below of = B22] {};
		\node[state] (B24) [below of = B23] {};
		\node[state] (B25) [below of = B24] {};
		
		\path[-]  (Root) edge (B11)
		(B11) edge (B12)
		(B12) edge (B13)
		;
		\path[-]  (Root) edge (B21)
		(B21) edge (B22)
		(B22) edge (B23)
		(B23) edge (B24)
		(B24) edge (B25)
		;	
		\node[] (L-B11) [above left of =B11, node distance = 0.5cm] {};	
		\node[] (L-B13) [below left of =B13,node distance = 0.5cm] {};	
		\node[] (R-B21) [above right of =B21, node distance = 0.5cm] {};	
		\node[] (R-B25) [below right of =B25,node distance = 0.5cm] {};		 	  
		\draw [-,thick, blue,decorate,decoration={brace,amplitude=2pt,mirror},xshift=-0pt,yshift=0](L-B11) -- (L-B13) node[left,black,midway,xshift=0] {$\ell_1$};
		\draw [-,thick, blue,decorate,decoration={brace,amplitude=2pt,mirror},xshift=0pt,yshift=0](R-B25) -- (R-B21) node[right,black,midway,xshift=0.1cm] {$\ell_2$};			 	  		  
		
		\node[state] (2-Root) [right of = Root, node distance =6cm,fill = gray] {};
		\node[state] (N11) [below left of = 2-Root, pattern = north west lines,pattern color = red] {};
		\node[state] (N12) [below of = N11, pattern = north west lines,pattern color = red] {};
		\node[state] (N13) [below of = N12, pattern = north west lines,pattern color = red] {};
		
		\node[state] (N21) [below right of = 2-Root,fill = gray] {};
		\node[state] (N22) [below of = N21,fill = gray] {};
		\node[state] (N23) [below of = N22,fill = gray] {};
		\node[state] (N24) [below right of = N23] {};
		\node[state] (N25) [below of = N24] {};
		\node[state] (N26) [below of = N25] {};
		
		\node[] (Ninv) [below left of =N23] {};
		\draw[-,dashed](N23) -- (Ninv);
		
		\node[] (textNewRoot) [right of = N23, node distance=1.8cm] {new root};
		\path (textNewRoot) edge (N23);
		
		\path[-]  (2-Root) edge (N11)
		(N11) edge (N12)
		(N12) edge (N13)
		;
		\path[-]  (2-Root) edge (N21)
		(N21) edge (N22)
		(N22) edge (N23)
		(N23) edge (N24)
		(N24) edge (N25)
		(N25) edge (N26)
		;	
		
		\node[] (R-N24) [above right of =N24, node distance = 0.5cm] {};	
		\node[] (R-N26) [below right of =N26,node distance = 0.5cm] {};		 	  
		\draw [-,thick, blue,decorate,decoration={brace,amplitude=2pt,mirror},xshift=0pt,yshift=0](R-N26) -- (R-N24) node[right,black,midway,xshift=0.1cm] {$s_1$};	
		
		\node[] (refFirstTree)  [right of =B22,node distance=0.7cm] {};
		\node[] (ref2ndTree)  [left of =N12,node distance=0.7cm] {};
		\path (refFirstTree) edge node{\small Capitulation of $1$}node[below,yshift=-0.3pt]{\small $2$ won the mining race}(ref2ndTree); 
		
		\node[] (MB1) [right of = B13,node distance=0.5cm] {\small $B_1$};
		\node[] (MB2) [left of = B25,node distance=0.5cm] {\small $B_2$};
		\node[] (NB1) [left of = N23,node distance=0.5cm] {\small $B_1$};
		\node[] (NB2) [left of = N26,node distance=0.5cm] {\small $B_2$};
	\end{tikzpicture}}	
	\end{minipage}	
	\caption{On the left, representation of the public tree $T$. The mining block of player $1$, $B_1$, must be the last block in the left-side branch; the mining block of player $2$, $B_2$, must be the last block in the right-side branch.
	On the right, capitulation of player $1$. Player $2$ won the mining race getting a branch of length $\ell_2+1$. Player $1$ validates $\ell_2+1-s_1$ blocks of player $2$; the new state is $(0,s_1)$.
	}\label{fig:Capitulation}
\end{figure}

This structure implies that each player takes part of a Markov Decision Process: At each turn, the game will be at a state $(\ell_1,\ell_2)$, and each player must decide if he will capitulate or not if he loses the mining race. 
When both decisions are taken, only two possible new states are reachable for the next turn: the one given by player $1$ winning the mining race (with probability $p_1$), and the one given by player $2$ winning the mining race (with probability $p_2$). 
By the Markovian property, for a state $(\ell_1,\ell_2)$, each player should make the same decision each time the game passes through that configuration.
Moreover, we take the values of $s_1$ and $s_2$ to be constant, both independent of the previous state (i.e. each time that a player capitulates, the player capitulates to the same state).
Thus, the strategies for a player can be summarized as what we call \emph{capitulation policies}.

\begin{definition}\label{def:CapPolicy} A capitulation policy for a player $i\in \{1,2\}$ is a couple $(C,s)$ such that the following holds:
	\begin{enumerate}[label=(\roman*)]
		\item $C: \{0,\ldots,d\}\times\{0,\ldots,d\}\to \{0,1\}$ and $s\in \{0,\ldots,d\}$,
		\item If $C(\ell_1,\ell_2) = 0$, it means that if a turn starts at state $(\ell_1,\ell_2)$, player $i$ continues to mine his or her branch, regardless if he or she wins the mining race or not that turn.
		\item If $C(\ell_1,\ell_2) = 1$, it means that if a turn starts at state $(\ell_1,\ell_2)$, player $i$ continues to mine his or her branch only if he or she wins the mining race at that turn, and he or she capitulates with $s_i=s$ otherwise.
	\end{enumerate}
\end{definition}

\subsection{Description of the Markov Chain}\label{subsec:Description}

In what follows consider $(C_1,s_1)$ and $(C_2,s_2)$ two feasible fixed capitulation policies, for player $1$ and $2$ respectively. 
Then, the Markov decision process previously described induces a Markov chain $(X_n)_{n\in\N}$, formally described as follows.
The states of the chain are given by $\mathcal{M}$, which is the subset of all states $(\ell_1,\ell_2)$ reachable from $(0,s_1)$ and $(s_2,0)$. For any initial distribution $\mu$ over $\mathcal{M}$, we consider the probability space $(\Omega,\mathcal{F},\P_{\mu})$ where $\Omega = \mathcal{M}\times \Omega_0 = \mathcal{M}\times (\{1,2\}\times\R_+)^{\N}$, the measurable sets are $\mathcal{F} = \mathcal{P}(\mathcal{M})\times \mathcal{F}_0$ and for every $(m,\omega)\in \Omega$, $m$ stands for the initial state and $\omega=(\omega_n,t_n)_{n\in\N}$ is the outcome of a sequence of mining races with their time lengths, and $\P_{\mu} = \mu\times \P_0$, that is
$\P_{\mu}(X_0 = m) = \mu(m)$ (see Section \ref{sec:Pre} for the definition of $(\Omega_0,\mathcal{F}_0,\P_0)$).
For a state $(\ell_1,\ell_2)\in \mathcal{M}$ there are two outgoing transitions, which are given by the following four cases:
\begin{enumerate}[label=(\alph*)]
	\item When $C_1(\ell_1,\ell_2)= 0$ and $C_2(\ell_1,\ell_2)= 0$, then $\P_{\mu}(X_{n+1}=(\ell_1+1,\ell_2)| X_n= (\ell_1,\ell_2)) = p_1$ and
	$\P_{\mu}(X_{n+1}=(\ell_1,\ell_2+1)| X_n = (\ell_1,\ell_2)) = p_2$.	
	\item When $C_1(\ell_1,\ell_2)= 1$ and $C_2(\ell_1,\ell_2)= 0$, then $\P_{\mu}(X_{n+1}=(\ell_1+1,\ell_2)| X_n= (\ell_1,\ell_2)) = p_1$ and $\P_{\mu}(X_{n+1}=(0,s_1)| X_n = (\ell_1,\ell_2)) = p_2$.        
	\item When $C_1(\ell_1,\ell_2)= 0$ and $C_2(\ell_1,\ell_2)= 1$, then $\P_{\mu}(X_{n+1}=(s_2,0)| X_n= (\ell_1,\ell_2)) = p_1$ and $\P_{\mu}(X_{n+1}=(\ell_1,\ell_2+1)| X_n = (\ell_1,\ell_2)) = p_2.$
	\item When $C_1(\ell_1,\ell_2)= 1$ and $C_2(\ell_1,\ell_2)= 1$, then $\P_{\mu}(X_{n+1}=(s_2,0)| X_n= (\ell_1,\ell_2)) = p_1$ and
	$\P_{\mu}(X_{n+1}=(0,s_1)| X_n = (\ell_1,\ell_2)) = p_2$.
\end{enumerate}
For each player $i\in \{1,2\}$, the expectations in the objective functions $R_i$ and $G_i$ are taken with respect to the probability space $(\Omega,\mathcal{F},\P_{\mu})$, that is,
\[
R_i = \lim_{n\to\infty}\E_{\mu}\left( \frac{\sum_{k=0}^{n-1} r_{i,k} - c_i\tau_k}{n} \right)\quad\mbox{ and }\quad G_i = \lim_{n\to\infty}\E_{\mu}\left( \frac{\sum_{k=0}^{n-1} r_{i,k}}{1+\sum_{k=0}^n\dist_k} \right),
\]
to emphasize the initial distribution whenever it is necessary. At turn $n$, the values of $r_{i,n}$ are computed depending on the transition $e =X_nX_{n+1}$ from the state $X_n$ to $X_{n+1}$. For each transition (edge) $e = (\ell_1,\ell_2)\to (\ell_1',\ell_2')$, the rewards of each player are the following:
If $(\ell_1',\ell_2') = (0,s_1)$, then $r_1(e) = 0$ and $r_2(e) = \ell_2+1 - s_1$.
If $(\ell_1',\ell_2') = (s_2,0)$, then $r_1(e) =\ell_1+1-s_1 $ and $r_2(e) = 0$.
In any other case, $r_1(e) = r_2(e) = 0$. Then,
$r_{i,n} = r_i(X_nX_{n+1})\text{ for each player } i\in\{1,2\}\text{ and every } n\in \N$.
We denote by $P$ the transition matrix of this chain, omitting the dependence on $(C_1,s_1,C_2,s_2)$ when there is no ambiguity. 
As usual, if the initial distribution $\mu$ is the delta distribution $\delta_m$ for some state $m\in \mathcal{M}$, we will simply write $\P_m$ and $\E_m$ instead of $\P_{\delta_m}$ and $\E_{\delta_m}$ in this case. 
\begin{lemma}\label{lemma:irreducible} 
	The chain $(\mathcal{M},P)$ is irreducible and there is a unique stationary distribution $\pi_P$.
\end{lemma}

\begin{proof}
By construction, each state of $\mathcal{M}$ is either reachable from $(0,s_1)$ or $(s_2,0)$. Furthermore, for each state $(\ell_1,\ell_2)\in \mathcal{M}$, there is at least one path to arrive to $(s_2,0)$, given by consecutive winnings of $P_1$ of the mining races. Similarly, consecutive winnings of $P_2$ form a path from $(\ell_1,\ell_2)$ to $(0,s_1)$. Thus, the chain is irreducible.
Since the chain is finite, the existence and uniqueness of the stationary distribution follows (see e.g. \cite[Corollary 1.17]{LevinPeres2017}). 
\end{proof}

For the chain $(\mathcal{M},P)$ let us define the sets
\begin{equation}\label{eq:Boundaries}
	\begin{aligned}
		\partial_1\mathcal{M} &= \left\{ m\in\mathcal{M}: P(m, (s_2,0))>0 \right\},\\
		\partial_2\mathcal{M} &= \left\{ m\in\mathcal{M}: P(m, (0,s_1))>0 \right\}.
	\end{aligned}
\end{equation}
The set $\partial_i\mathcal{M}$ corresponds to the states $m\in \mathcal{M}$ for which $C_{-i}(m) = 1$, that is, the set of states for which player $i$ wins the current round after winning the current mining race.  Let us define the $|\mathcal{M}|$-dimensional vectors $\hat{r}_1$ and $\hat{r}_2$ as follows:
$\hat{r}_1(\ell_1,\ell_2) = p_1(\ell_1+1-s_2)$ if $(\ell_1,\ell_2)\in\partial_1\mathcal{M}$, and $\hat{r}_1(\ell_1,\ell_2) = 0$ otherwise;
$\hat{r}_2(\ell_1,\ell_2) = p_2(\ell_2+1-s_1)$ if $(\ell_1,\ell_2)\in \partial_2\mathcal{M}$, and $\hat{r}_2(\ell_1,\ell_2) = 0$ otherwise.
In principle, the values of the objective functions might depend on the initial distribution $\mu$. 
However, the ergodic theorem (see e.g. \cite[Theorem 1.10.2]{Norris1998}) suggest that regardless the initial distribution, the values of $R_i$ and $G_i$ should depend only on the invariant distribution $\pi_P$, for each $i\in \{1,2\}$. 
The following proposition formalizes this notion.  
Given two vectors $x,y$ with entries in $\mathcal{M}$, we denote by $\langle x,y\rangle=\sum_{a\in \mathcal{M}}x(a)y(a)$ the inner product between $x$ and $y$.

\begin{proposition}\label{prop:ConvergenceAverage} For the chain $(\mathcal{M},P)$ and any initial distribution $\mu$, we have that 
	$G_i = \langle \pi_P, \hat{r}_i\rangle/\dist$,
	for each $i\in \{1,2\}$, where 
	\[\dist = \lim_{n\to\infty} \frac{1}{n}\E_{\mu}\left(1+\sum_{k=1}^n \dist_k\right).\] 
	Furthermore,
	$\dist = \langle \pi_P, \hat{r}_1 + \hat{r}_2\rangle$ and $R_i = \dist(G_i - c_i\overline{\tau}) = \langle \pi_P,\hat{r}_i\rangle - c_i\overline{\tau} \dist$ for each $i\in \{1,2\}$.
\end{proposition}

\begin{proof}
Let $E$ be the set of all edges of the chain $(\mathcal{M},P)$, that is, $e =(m_1m_2)\in E$ if $P(m_1,m_2)>0$. Consider $i\in \{1,2\}$ and let $(Z_n)_{n\in\N}$ be the stochastic process over $(\Omega,\mathcal{F},\P_{\mu})$ given by $Z_n(m,\omega) = (X_{n}(m,\omega)X_{n+1}(m,\omega))\in E$.
Then, the process $(Z_n)_{n\in\N}$ is a Markov chain $(E,Q)$ with initial distribution $\tilde{\mu}$, where
\[
\tilde{\mu}(m_1,m_2) = \begin{cases}
	\mu(m_1)\cdot p_1\qquad&\mbox{ if }m_2\mbox{ follows after player }1\mbox{ won the mining race,}\\
	\mu(m_1)\cdot p_2\qquad&\mbox{ if }m_2\mbox{ follows after player } 2\mbox{ won the mining race.}
\end{cases}
\]
The transition matrix $Q$ over $E$ is given as follows: For $e=(m_1m_2)$ and $e'=(m_1'm_2')$, if $m_2 = m_1'$, then $Q(e,e') = P(m_1',m_2')$; otherwise, $Q(e,e')=0$.
We have that $(E,Q)$ is also irreducible and finite and therefore it has a unique stationary distribution $\pi_Q$ \cite[Corollary 1.17]{LevinPeres2017}.
Furthermore, for each turn $n$, we have that $r_{i,n} = r_i(Z_n)$ and $\dist_n=r_1(Z_n) + r_2(Z_n)$. Since every finite irreducible Markov chain is also positive recurrent (see e.g. \cite[Theorem 1.7.7]{Norris1998}) we can apply the ergodic theorem \cite[Theorem 1.10.2]{Norris1998} obtaining that
\footnote{A sequence of random variables $(Y_n)_{n\in \N}$ converges in probability to a random variable $Y$, denoted by $Y_n\xrightarrow{\P} Y$, if for every $\varepsilon>0$ we have that $\lim_{n\to \infty}\P(|Y_n-Y|>\varepsilon)=0$.}
\begin{align*}
	f_{1,n}=\frac{1}{n}\sum_{k=0}^{n-1} r_1(Z_k) &\xrightarrow{\P_{\mu}} \langle \pi_Q,r_1\rangle,\\
	f_{2,n}=\frac{1}{n}\sum_{k=0}^{n-1} r_2(Z_k) &\xrightarrow{\P_{\mu}} \langle \pi_Q,r_2\rangle,\\
	h_{n}=\frac{1}{n}\left(1+\sum_{k=0}^{n-1}\dist_k\right) &\xrightarrow{\P_{\mu}} \langle \pi_Q,r_1+r_2\rangle.
\end{align*}

On the one hand, the total amount of mined blocks at turn $n$ (validated or not) is $n$. On the other hand, in the worst case scenario, there is at least one block that is validated every $2d$ turns. Thus, we get that
\[
\frac{n}{2d}-1\leq\left\lfloor \frac{n}{2d}\right\rfloor\leq \sum_{k=0}^{n-1}\dist_k\leq n.
\]
The upper bound yields that $\{f_{1,n}\}_{n\in \N}$, $\{f_{2,n}\}_{n\in \N}$ and $\{h_n\}_{n\in \N}$ are uniformly bounded by two and therefore we have the convergence in expectation:\footnote{If we have a sequence of random variables $(Y_n)_{n\in \N}$ that converges in probability to a random variable $Y$ and such that $|Y_n|\le C$ for some $C$ and every $n\in \N$, then $\mathbb{E}(Y_n)\to \mathbb{E}(Y)$.} $\E_{\mu}(f_{i,n})\to \langle \pi_Q,r_i\rangle$ for each $i\in \{1,2\}$ and $\E_{\mu}(h_n) \to \langle \pi_Q,r_1+r_2\rangle$ when $n\to \infty$. 
Furthermore, we have that the ratio of the sequences converges in probability,\footnote{If we have two sequences of positive random variables $\{Y_n\}_{n\in \N}$ and $\{Z_n\}_{n\in \N}$ that converge in probability to $Y$ and $Z$ respectively, with $\{Z_n\}_{n\in \N}$ bounded away from zero, then $\{Y_n/Z_n\}_{n\in \N}$ converges in probability to $Y/Z$.} that is
\[
\frac{f_{i,n}}{h_n}\xrightarrow{\P_{\mu}}\frac{\langle\pi_Q,r_i\rangle}{\langle \pi_Q,r_1+r_2\rangle}.
\]
Finally, recalling that $\dist_k = r_1(Z_k) + r_2(Z_k)$, we have that $f_{i,n}/h_n\leq 1$ for every positive integer $n$, and therefore we have the convergence in expectation, 
\[
\lim_{n\to \infty}\E_{\mu}\left( \frac{f_{i,n}}{h_{n}} \right)=\frac{\langle \pi_Q,r_i\rangle}{\langle \pi_Q,r_1+r_2\rangle},
\]
for each $i\in \{1,2\}$.
For every $e=(m_1m_2)\in E$ we have that
\[
\pi_Q(e) = \begin{cases}
	\pi_P(m_1)\cdot p_1\qquad&\mbox{ if }m_2\mbox{ follows after player 1 won the mining race,}\\
	\pi_P(m_1)\cdot p_2\qquad&\mbox{ if }m_2\mbox{ follows after player 2 won the mining race.}
\end{cases}
\]
Now, noting for $e=(m_1m_2)$ we have that that $r_1(e)= 0$, whenever $m_2\neq (s_2,0)$ we have that
\begin{align*}
	\langle \pi_Q, r_1\rangle &= \sum_{m\in \mathcal{M}} \pi_Q(m, (s_2,0))r_1(m,(s_2,0))\\
	&= \sum_{m\in \partial_1\mathcal{M}} \pi_P(m)p_1r_1(m,(s_2,0))=\sum_{m\in\mathcal{M}} \pi_P(m)\hat{r}_1(m) = \langle \pi_P,\hat{r}_1\rangle.
\end{align*}
Similarly, we get $\langle \pi_Q,r_2\rangle = \langle \pi_P,\hat{r}_2\rangle$ and $\langle \pi_Q,r_1 +r_2\rangle = \langle\pi_P,\hat{r}_1+\hat{r}_2\rangle$.
The proof is finished noting that $\dist = \lim_{n\to \infty}\E_{\mu}(h_n)$, that $G_i = \lim_{n\to \infty} \E_{\mu}(f_{i,n}/h_n)$ for each $i\in \{1,2\}$, and writing
\begin{align*}
	R_i &= \lim_{n\to \infty}\E_{\mu}\left(f_{i,n} - \frac{1}{n}\sum_{k=0}^{n-1}c_i\tau_k\right)\\
	&=\lim_{n\to \infty} \E_{\mu}\left(f_{i,n} - c_ih_n\frac{\sum_{k=0}^{n-1}\tau_k}{1+\sum_{k=0}^{n-1}\dist_k}\right)\\
	&=\lim_{n\to \infty} \E_{\mu}(f_{i,n}) - c_i\E(h_n)\cdot \E_{\mu}\left(\frac{\sum_{k=0}^{n-1}\tau_k}{1+\sum_{k=0}^{n-1}\dist_k}\right) = \langle\pi_P,\hat{r}_i\rangle - c_i\overline{\tau}\dist. \qedhere
\end{align*}
\end{proof}
The formula $R_i = \dist(G_i - c_i\overline{\tau})$  reflects the trade-off between $R_i$ and $G_i$: While strategic mining might increase the value of $G_i$ above $p_i$, it does that by reducing $\dist$ (which is always equal to 1 if both players play Frontier), and thus, since $G_i \cdot \dist = \langle \pi_P,\hat{r}_i\rangle \leq p_i$ (see the corollary below) the true effect of strategic mining is not in the revenues per turn, but in the averaged costs. 

\begin{corollary}\label{cor:EqWithoutCost} Regardless the capitulation policies $(C_1,s_1)$ and $(C_2,s_2)$, one always has that
	$\langle \pi_P,\hat{r}_i\rangle \leq p_i,\mbox{ for }i\in\{1,2\}$.
	Therefore, if the normalized marginal costs $c_i$ are zero and both players try to maximize their asymptotic expected revenue, $R_1$ and $R_2$, then Frontier is a Nash equilibrium.
\end{corollary}
\begin{proof}
Let us assume that player $2$ is playing Frontier. 
Regardless the capitulation policy $(C_1,s_1)$ of player $1$, in the best case player $1$ will get the reward of all the blocks that player $1$ has mined, and thus
$\sum_{k=0}^{n-1}r_{1,k} \leq \sum_{k=0}^{n-1}1_{\{\omega_k=1\}}$. This yields that
\begin{align*}
	R_1 = \lim_{n\to\infty}\frac{1}{n}\E\left(\sum_{k=0}^{n-1}r_{1,k}\right)\leq \lim_{n\to\infty}\frac{1}{n}\E\left(\sum_{k=0}^{n-1}1_{\{\omega_k=1\}}\right) = p_1. 
\end{align*}
Since $p_1$ is the value of $R_1$ under Frontier, the conclusion follows. 
\end{proof}
The above corollary shows the impact of the selection in the payoff functions. Our result does not, in principle, contradict the results obtained in \cite{eyal2014majority,Koutsoupias2016,marmolejo2019competing}, since they assume that players aim to maximize the ratio of owned validated blocks, that is, $G_i$. However, when we look revenues, Proposition \ref{prop:ConvergenceAverage} and Corollary \ref{cor:EqWithoutCost}  tell us that maximizing $G_i$ is not necessarily optimal: informally, in simple words, strategic players maximizing $G_i$ get a bigger portion of a smaller cake. Furthermore, the same proof of Corollary \ref{cor:EqWithoutCost}  is valid without the immediate release assumption, since it is based on the fact that under strategic mining, it is not possible for a player to win more blocks than those that the player has found, which is the payment when all players play Frontier. 


\section{Convergence Study for the Markov Chain\label{sec:Times}}

In what follows we suppose that the player $2$ plays the Frontier strategy, that is, 
$C_2 = 1_{\{(\ell_1,\ell_2):\ell_1\geq \ell_2\}}$ and $s_2 = 0$. 
We denote $s = s_1$ and $C = C_1$ the capitulation policy for player 1. With this notation we can describe the sets in \eqref{eq:Boundaries} as $\partial_1 \M = \{(\ell,\ell) \in \M\}$ and $\partial_2 \M = \{(\ell_1,\ell_2) \in \M: C(\ell_1,\ell_2) = 1\}$.
That is, $\partial_i \M$ is the set of states where the other player capitulates if player $i$ wins the current mining race. We will call these states and the corresponding outgoing transitions, \emph{capitulation states} and \emph{transitions} for player $i$.
Any other state or transition is called \emph{interior}.
Assuming rationality of players, it is natural to consider capitulation policies $(C,s)$ satisfying the following:
\begin{equation}\label{eq:RationalPolicies}
	\text{When }C(\ell_1,\ell_2) = 0, \text{ we have that } C(\ell_1,\ell_2') = 0\mbox{ for all }\ell_2'\leq \ell_2.
\end{equation}
The above implication tells us that if the strategic player is willing to continue the round at state $(\ell_1,\ell_2)$, then he or she should be willing to continue the round for any other state $(\ell_1,\ell_2')$ with $\ell_2'\leq \ell_2$, since those are more favorable states than $(\ell_1,\ell_2)$. Thus, it is natural to introduce the notion of \emph{gap tolerance}, which should be the maximum value of $\ell_2$ satisfying that $C(\ell_1,\ell_2) = 0$.

\begin{definition}\label{def:GapTol} For a capitulation policy $(C,s)$ for player $1$, we define the gap tolerance as the function
	$g:\{0,\ldots,d\}\to \{0,\ldots,d\}$
	such that $g(\ell_1)= \max\{ \ell_2-\ell_1\ :\ C(\ell_1,\ell_2) = 0 \}$.
	We define the maximum gap tolerance as $\overline{g} = \max_{\ell_1\in \{0,\ldots,d\}} g(\ell_1)$.
\end{definition}

Note that we can always characterize the maximal gap tolerance by
$\overline{g} = \max_{(\ell_1,\ell_2)\in\mathcal{M}}(\ell_2-\ell_1)$,
where the inequality holds by the construction of the reachable states.
In order to study the Markov chain, it results useful to consider a lattice representation of $(\mathcal{M},P)$.
In this representation, we represent the states of the Markov chain in the two dimensional integer lattice and we represent the transitions as arrows.
See Figure \ref{fig:lat-rep} for an example.

\begin{figure}[h!t]
	\centering
	\includegraphics[width=.6\textwidth,height=.3\textheight]{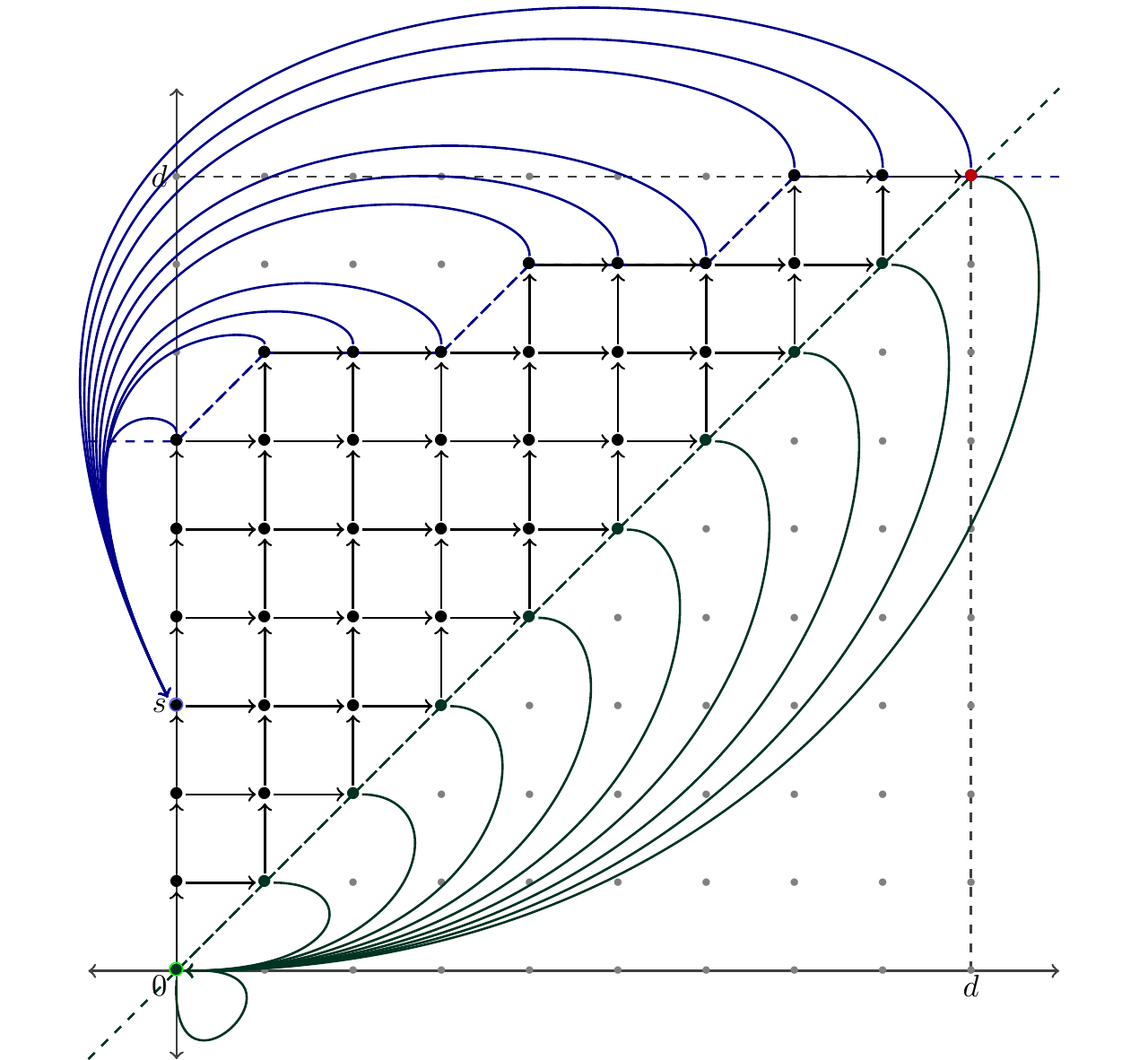}
	\caption{
		Lattice representation of our Markov chain.
		Blue dots and arrows correspond to capitulation states and transitions for player $1$ and green ones, for player $2$. Black arrows correspond to interior transitions.
	}
	\label{fig:lat-rep}
\end{figure}

If we have a state $(\ell_1,\ell_2)\in \partial_1\mathcal{M}$, we identify
$(\ell_1+1,\ell_2)$ with $(0,0)$. 
If we have $(\ell_1,\ell_2)\in \partial_2\mathcal{M}$, we identify
$(\ell_1,\ell_2+1)$ with $(0,s)$.
Observe that, regardless the capitulation policy of player $1$, one has that the induced Markov chain is aperiodic, since $P^n((0,0),(0,0))\geq p_1^n >0$, for every positive integer $n$.
Thus, the classic convergence theorem (see e.g. \cite[Theorem 4.9]{LevinPeres2017}) implies that
$\max_{m\in \mathcal{M}}\| P^n(m,\cdot) - \pi_P \|_{\TV}\to 0$,
where $\|\cdot\|_{\TV}$ is the total variation norm of measures. Together with Proposition \ref{prop:ConvergenceAverage}, this yields that, in the long run, the behavior of the mining game can be fully described by the stationary distribution $\pi_P$.  

\subsection{Mixing Time of the Markov Chain\label{subsec:MixingTime}}

For the policy $(C,s)$, we are interested in estimating how much time do we need in order to observe the objective values considered, $R_1$ the expected revenue, and $G_1$ the expected ratio of validated blocks. 
As a benchmark, we consider the mixing time of the induced Markov chain $(\mathcal{M},P)$.
Recall that the mixing time for a tolerance $\varepsilon>0$ is given by 
$t_{\mix}(\varepsilon) = \min\left\{ n\in \N \ :\ \max_{m\in\mathcal{M}}\|P^n(m,\cdot) -\pi_P\|_{\TV}\leq \varepsilon\right\}$.
In what follows we provide a bound for $t_{\mix}(\varepsilon)$ in terms of the maximum gap tolerance supported by player $1$.

\begin{theorem}  
	\label{thm:mixing}
	Let $(C,s)$ be the capitulation policy of $1$ and let $(\mathcal{M},P)$ the associated Markov chain when the player 2 plays Frontier. Then, for every $\varepsilon\in (0,1)$, we have that
	\[
	t_{\mix}(\varepsilon) \leq \left\lceil\frac{\ln(\varepsilon)}{\ln\Big(1 - p_1^{\overline{g}+1}\Big)}\right\rceil(\overline{g}+1).
	\]
\end{theorem}

\begin{proof}
Choose any two states $x,y\in \mathcal{M}$ and let $(Z_n)_{n\in\N}$ be a sequence of independent Bernoulli trials of parameter $p_1$. 
If $Z_n$ is a success, it represents that player $1$ has won the $n$-th mining race. 
Let $(A_n,B_n)_{n\in\N}$ be the stochastic process with values in $\mathcal{M}\times\mathcal{M}$ given by
\begin{enumerate}
	\item $A_0\sim \delta_x$,
	\item $B_0\sim \delta_y$,
	\item $A_{n+1} = A_n + (Z_n,1-Z_n)$ and $B_{n+1} = B_n + (Z_n,1-Z_n)$.
\end{enumerate}
Following the notation of \cite[Chapter 5]{LevinPeres2017}, we will consider $\P_{x,y}$ as the probability measure over a space where the random variables $A_n,B_n$ and $Z_n$ are defined, and satisfying that $A_0\sim \delta_x$, $B_0\sim \delta_y$ and $Z_n\sim B(p_1)$.
By construction, for every states $(\ell_1,\ell_2), (\ell_1',\ell_2'), (\ell_1'',\ell_2'')\in\mathcal{M}$ we have that
\begin{align*}
	\P_{x,y}(A_{n+1}=(\ell_1'',\ell_2'')| A_n = (\ell_1,\ell_2), B_n = (\ell_1',\ell_2')) &= P((\ell_1,\ell_2), (\ell_1'',\ell_2'')),\\
	\P_{x,y}(B_{n+1}=(\ell_1'',\ell_2'')| A_n = (\ell_1,\ell_2), B_n = (\ell_1',\ell_2')) &= P((\ell_1',\ell_2'), (\ell_1'',\ell_2'')).
\end{align*}
This yields that $(A_n,B_n)$ is a coupling for the Markov chain $(\mathcal{M},P)$, satisfying that $A_n = B_n$ implies $A_k = B_k$ for every $k\geq n$.
Intuitively, for each transition $n$ to $n+1$, either both processes $A_n$ and $B_n$ move up  or both processes move to the right in the lattice representation of Figure \ref{fig:lat-rep}.
Let $T_c$ be the \emph{coalescence time} of the coupling, that is, $T_c = \min\{n \in \N :\ A_n = B_n \}$.
Note that, for any turn $n$, and any state $(\ell_1,\ell_2)\in\mathcal{M}$, we have that
\begin{align*}
	\P_{x,y}(A_{n+g+1} = (0,0)| A_n = (\ell_1,\ell_2), Z_{n}=1,\ldots,Z_{n+g}=1) &= 1,\\
	\P_{x,y}(B_{n+g+1} = (0,0)| A_n = (\ell_1,\ell_2), Z_{n}=1,\ldots,Z_{n+g}=1) &= 1.
\end{align*}
The above equations follow from the fact that if $(\ell_1,\ell_2)\in \mathcal{M}$, then $0\leq \ell_2-\ell_1\leq \overline{g}$. 
Thus, after $\overline{g}+1$ consecutive winnings of player $1$, the state must be $(0,0)$. 
Indeed, on the one hand, if $\ell_2-\ell_1=\overline{g}$, it means that after $g$ consecutive wins of player $1$, the chain is at $(\ell_2,\ell_2)\in \partial_1\mathcal{M}$, and thus, after one more win of player $1$, the chain goes to $(0,0)$. 
On the other hand, if the chain goes back to $(0,0)$ before the $\overline{g}+1$ consecutive wins, each win of player $1$ maintains the chain at $(0,0)$.

For every positive integer $n$, consider the event $S_n$ where the sequence $(Z_0,\ldots,Z_{n-1})$ contains $\overline{g}+1$ consecutive successes.
Then, we have that $\P_{x,y}(S_n)\leq\P_{x,y}(T_c \leq n)$.
While $S_n$ is a well-known event in the literature, it is hard to explicitly estimate its probability. 
Thus, we will provide a simpler lower bound for it. 
Let $U_k$ be the event given by $Z_{k(\overline{g}+1)}=1,\ldots,Z_{(k+1)(\overline{g}+1)-1}=1$.
We have $\P_{x,y}(U_k) = p_1^{\overline{g}+1}$ for each $k\in \N$ and the sequence $(U_k)_{k\in\N}$ is independent. 
Furthermore, we have that
\begin{align*}
	\P_{x,y}(T_c\leq n(\overline{g}+1))&\geq \P_{x,y}\Big(S_{n(\overline{g}+1)}\Big)\geq \P_{x,y}\left(\bigcup_{j=0}^{n-1}U_j\right)= 1-\left(1-p_{1}^{\overline{g}+1}\right)^n.
\end{align*}
We deduce that $\P_{x,y}(T_c > n(\overline{g}+1))\leq (1-p_{1}^{\overline{g}+1})^n$. 
On the other hand, we have that $(1-p_{1}^{\overline{g}+1})^n\leq \varepsilon$ if and only if $n\geq \ln(\varepsilon)/\ln(1-p_1^{\overline{g}+1})$,
and then we conclude that $\P_{x,y}(T_c > \overline{n}(\overline{g}+1)) \leq \varepsilon$, for 
\[\overline{n} = \left\lceil\frac{\ln(\varepsilon)}{\ln(1-p_1^{\overline{g}+1})}\right\rceil\] 
and so by \cite[Corollary 5.5]{LevinPeres2017}, we deduce that $t_{\mix}(\varepsilon) \leq \overline{n}(\overline{g}+1)$,
finishing the proof. 
\end{proof}

As a direct corollary from Theorem~\ref{thm:mixing}, when the time horizon $T$ and the total variation tolerance $\varepsilon$ are given, we can provide a lower bound on $p_1$ in terms of these parameters and the maximum gap tolerance $\overline{g}$.

\begin{corollary}
	\label{coro:p-mix}
	For every $\epsilon > 0$ and $T > 0$, we have $t_{\mix}(\epsilon) \leq T$ when 
	$p_1 \geq \widetilde{p}_1 = \Big(1 - \epsilon^{\frac{\overline{g}+1}{T-(\overline{g}+1)}}\Big)^{\frac{1}{\overline{g} + 1}}$.
\end{corollary}

\begin{table}[h]
	\centering
	\begin{tabular}{l|*{10}{|r}}
		$\overline{g}$ & $1$ & $2$ & $3$ & $4$ & $5$ & $6$ & $7$ & $8$ & $9$ & $10$ \\
		\hline
		$\widetilde{p}_1$ &
		$0.037$ &
		$0.127$ &
		$0.229$ &
		$0.322$ &
		$0.401$ &
		$0.467$ &
		$0.522$ &
		$0.569$ &
		$0.609$ &
		$0.642$
	\end{tabular}
	\caption{Examples of rapid mixing strategies for $\epsilon=10^{-3}$ and $T = 10^4$. For any $(C,s)$ with the given maximum gap tolerance $\overline{g}$, the associated chain satisfies $t_{\mix}(\epsilon) < T$ whenever $p_1\geq \widetilde{p}_1$.}
	\label{tab:rapid-mixing}
\end{table}

\begin{figure}[h!t]
	\centering
	\includegraphics[scale=1.2]{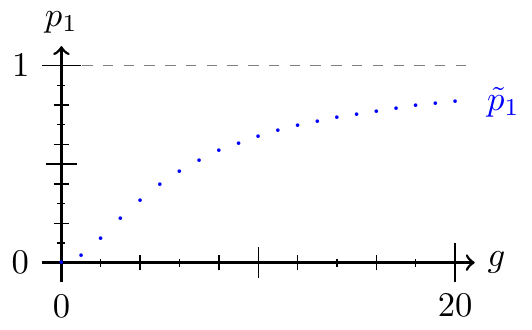}
	\hspace{4em}
	\includegraphics[scale=1.2]{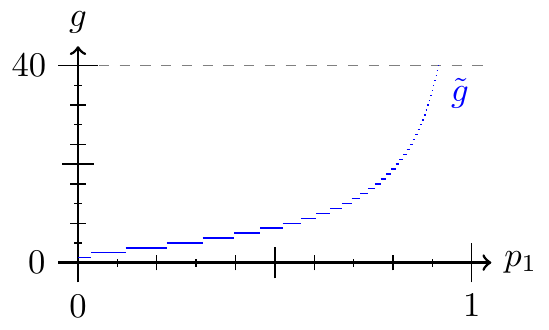}
	\caption{On the left, rapid mixing strategies for $\epsilon=10^{-3}$ and $T = 10^4$ with $\bar g\in \{1,\ldots,20\}$.
	For any $(C,s)$ with the given maximum gap tolerance $\bar g$, the associated Markov chain mixes rapidly (that is, $t_{\mix}(\epsilon) < T$) whenever $p_1\geq \widetilde{p}_1$.
	On the right, rapid mixing strategies for $\epsilon=10^{-3}$ and $T = 10^4$. For a given $p_1$, the associated chain mixes rapidly whenever $\overline{g} < \widetilde{g}$.}
	\label{fig:rapid-mixing}
\end{figure}
For instance, for $\varepsilon = 10^{-3}$ and $T \leq 10^4$, that is, about 70 days, if the average time of a mining race is 10 minutes, we get the values of $\widetilde{p}_1$ as in Table~\ref{tab:rapid-mixing} and Figure~\ref{fig:rapid-mixing}.
The inverse problem of determining $\widetilde{g}$, for a given $p_1$, such that any capitulation policy with maximum  gap tolerance $\overline{g} \leq \widetilde{g}$ exhibits fast mixing, can be easily read from the previous analysis, since our upper bound for $t_{\mix}(\epsilon)$ is increasing in $\overline{g}$ (see Figure~\ref{fig:rapid-mixing}).
In the following theorem, we show the existence of a capitulation policy for which the mixing time of the associated Markov chain is exponentially large as a function of $d$. 

\begin{theorem}
	\label{thm:large-mixing}
	There exists a capitulation policy for player 1 for which the associated Markov chain has a mixing time of at least $\frac{1}{4}(5/4)^d$.
\end{theorem}

\begin{proof}
Consider the capitulation policy $(C,s)$ for player 1 given as follows: $s = d$, $C(\ell_1,\ell_2) = 1$ if $\ell_2 = d$, and zero otherwise.
We denote by $\pi$ the stationary distribution of the associated Markov chain $(\mathcal{M},P)$.
Let $S\subseteq \mathcal{M}$ a subset of states given by $S=\{ (\ell_1,\ell_2) : \ell_2\leq d-1 \}$.
In what follows we denote $\pi(S) = \sum_{(\ell_1,\ell_2)\in S} \pi(\ell_1,\ell_2)$ and $P^t((0,d), S) = \sum_{(\ell_1,\ell_2)\in S} P^t((0,d),(\ell_1,\ell_2))$.
Then, observe that for any positive integer $t$ we have
\begin{align*}
	\| P^t((0,d),\cdot) - \pi \|_{\TV} &\geq    \pi(S) - P^t((0,d), S)\\
	&= \pi(S) - \P(X_t \in S | X_0 = (0,d))\\
	&\geq \pi(S) - \P(X_1 \in S\vee\cdots\vee X_t\in S | X_0 = (0,d))\\
	&= \pi(S) - \P(\tau_{(0,d)}(S) \leq t)\\
	&= \pi(S) - \P(\tau_{(0,d)}(0,0) \leq t),
\end{align*}
where $\tau_{(0,d)}(S)$ and $\tau_{(0,d)}(0,0)$ denote the hitting times of $S$ and $(0,0)$ starting from $(0,d)$, respectively. 
Furthermore, $\tau_{(0,d)}(0,0)\leq t$ if and only if there is a sequence of $d+1$ consecutive winnings for player 1 within the $t$ turns. 
For every $i\in \{0,1,\ldots,t-d-1\}$, let $A_i$ be the event in which player 1 wins consecutively from turn $i$ to turn $i+d$.
Thus, by the union bound we have that
\begin{align*}
	\P(\tau_{(0,d)}(0,0)\leq t) &= \P\left(\bigcup_{i=0}^{t-d-1} A_i\right)\leq \sum_{i=0}^{t-d-1} \P(A_i)= (t-d)p^{d+1},
\end{align*}
where $p$ is the probability that player 1 is winning the mining race.
In particular, by setting $t = t_{\rm mix}(\varepsilon)$, we have that $\varepsilon\geq \pi(S) - (t_{\rm mix}(\varepsilon)-d)p^{d+1}$ and therefore
\[
t_{\rm mix}(\varepsilon) \geq \frac{\pi(S) - \varepsilon}{p^{d+1}} + d\geq \frac{\pi(S) - \varepsilon}{p^{d}}.
\]
Similarly, we have that
\begin{align*}
	\| P^t((0,0),\cdot) - \pi \|_{\TV} &\geq \pi(S^c) - P^t((0,0),S^c)\geq \pi(S^c) - \P(\tau_{(0,0)}(S^c) \leq t).
\end{align*}
We have that $\tau_{(0,0)}(S^c)\leq t$ if and only if there is a sequence of $d+z$ turns with $z\leq d$ within the $t$ turns starting from $(0,0)$ such that the following holds:
i) Player 2 wins $d$ mining races, ii) player 1 wins $z$ mining races and iii) there is no sub-sequence where player 1 has more wins than player 2.
For any such pair $(d,z)$, let $B_{d,z}$ the event in which the previous three conditions hold and let $C_{d,z}$ be the event in which only conditions i) and ii) hold.
Let us denote by $B_s$ the event in which there is a sequence of $d+z$ turns with $z\leq d$ within the $t$ turns starting from $(0,0)$ at turn $s\leq t-d$.
In particular, we have that $B_s=\cup_{i=0}^{d-1}B_{d,z}$. 
Then, we have that
\begin{align*}
	\P(B_s) &\leq \sum_{z=0}^{d-1} \P (B_{d,z})\leq\sum_{z=0}^{d-1} \P (C_{d,z}) = \sum_{z=0}^{d-1}\binom{d+z}{z}q^{d}p^z,
\end{align*}
Now, using the hockey-stick identity, we have that
\begin{align*}
	\P(B_s) &\leq \sum_{z=0}^{d-1}\binom{d+z}{z}q^{d}p^z\\
	&\leq q^d \sum_{z=0}^{d-1}\binom{d+z}{z}\\
	&= q^d\binom{2d}{d-1}= \frac{2d(2d-1)}{d+1}q^d\cdot \frac{1}{d}\binom{2d-2}{d-1}= \frac{2d(2d-1)}{d+1}q^d C(d-1)\leq \frac{d^2}{d+1}q^d 4^d,
\end{align*}
where $C(n)$ is the $n$-th Catalan number, which is 
bounded from above by $4^n$. 
Therefore, we have that
\begin{align*}
	\P(\tau_{(0,0)}(S^c) \leq t)&= \P\left( \bigcup_{s=0}^{t-d} B_s \right)\leq \sum_{s=0}^{t-d}\P(B_s)\leq \sum_{s= 0}^{t-d}\frac{d^2}{d+1}q^d 4^d\leq (t-d+1)\frac{d^2}{d+1}q^d 4^d.
\end{align*}
Thus, by setting $t= t_{\rm mix}(\varepsilon)$ we conclude that
\[
\varepsilon\geq \pi(S^c) - (t_{\rm mix}(\varepsilon)-d+1)\frac{d^2}{d+1}q^d 4^d,
\]
which we can rewrite as
\[
t_{\mix}(\varepsilon) \geq \frac{d+1}{d^2}\frac{(\pi(S^c)-\varepsilon)}{(4q)^d} +d-1 \geq \frac{\pi(S^c)-\varepsilon}{(4q)^d}.
\]
Since $\pi(S)+ \pi(S^c) = 1$, one of these values has to be greater than $1/2$ and therefore by taking $\varepsilon=1/4$, $p=4/5$ and $q = 1-p = 1/5$, we deduce that
\[
t_{\mix}(1/4) \geq \frac{1}{4}\left(\frac{5}{4}\right)^d.\qedhere
\]
\end{proof}

\subsection{Hitting Time of the State $(d,d)$}\label{subsec:Hitting}

In this section, we are interested to know whether, for a given strategy of player 1, the truncated model with maximum depth $d$ is reliable or not. In our model, parallel blocks can not be both validated, since whenever a player mines a path of $d+1$ blocks, the other one is forced to capitulate. 
If the strategic player actually faces the situation of potentially parallel validated blocks, he or she might not capitulate due to the high losses of doing so. 
This undesirable situation arises only when the Markov chain is at $(d,d)$.

Ideally, strategic miners should have capitulation policies for which the state $(d,d)$ is unreachable. 
However, as we can see in the lattice representation of Figure \ref{fig:lat-rep}, this is not the case in general. 
We propose an alternative criterion, which considers the hitting time of the state $(d,d)$. 
If the hitting time of this state is very large (i.e. beyond human scale), the associated strategy can be {\it safely} played in the truncated model.

For example, if the average time $\tau_b$ of a mining race is 5 minutes, a time horizon of $T=10^8$ corresponds to about a millennia. Therefore, any capitulation strategy for which the hitting time of $(d,d)$ is greater that $T$, can be considered as safe. 
For the Bitcoin time scale of 100 years, $T=10^8$ means that the complete sequential game must be played about 10 times before $(d,d)$ is hit. 
On the other hand, a time horizon of $T=10^4$ corresponds to about one month, and clearly they should be considered as \emph{unsafe}. 
We show that, under certain conditions, unsafe strategies exist for this time horizon.
The benchmark $\tau_b \geq 5$ minutes  follows from the fact that $\tau_b = \dist\overline{\tau}$ (which is easily deduced from Proposition \ref{prop:ConvergenceAverage}) and that in the examples we provide next, $\dist \geq 1/2$. 

Formally, when the chain starts at state $(0,0)$, the hitting time of the state $(d,d)$ corresponds to  $\mathcal{T}(d,d)=\mathbb{E}(\min\{n\in \N:X_n=(d,d)\})$. 
We say that a capitulation policy is safe for $T>0$ if $\mathcal{T}(d,d)\geq T$.
Recall that a round of the game is a set of transitions from a capitulation transition (or the beginning of the game) to the next capitulation transition.
The length of a round, that is, the number of transitions between two capitulation transitions, is bounded above by $2d$, which is the number of interior transitions needed to go from $(0,0)$ to $(d,d)$.
Thus, if $\calR$ is the number of rounds before hitting $(d,d)$ for the first time, we have
\[
\E(\calR) \leq \mathcal{T}(d,d)\leq 2d \cdot \E(\calR).
\]
Then, up to a constant factor of $d$, studying $\mathcal{T}(d,d)$ can be reduced to estimating the value $\E(\calR)$.

\begin{proposition}
	\label{prop:E(R)}
	Let $N(\ell_1,\ell_2)$ be the number of interior lattice paths from $(\ell_1,\ell_2)$ to $(d,d)$ in the lattice representation of $(\M,P)$.
	Then 
	\[\frac{N(0,s) p_2^{2s}}{N(0,0)^2p_1^dp_2^d} \leq \E(\calR) \leq \frac{N(0,0)}{N(0,s)^2p_1^dp_2^{d+s}}.\]
\end{proposition}
\begin{proof}
In the $n$-th round, the probability of hitting $(d,d)$ is given by
\[
\rho_n = 
\begin{cases}
	\nu_n(0,0) N(0,0) p_1^dp_2^d + \nu_n(0,s) N(0,s) p_1^dp_2^{d-s} & \text{ if } s > 0, \\
	\nu_n(0,0) N(0,0) p_1^dp_2^d & \text{ if } s = 0,
\end{cases}
\]
where $\nu_n(\ell_1,\ell_2)$ is the probability that the round starts at $(\ell_1,\ell_2)$.
In particular, we have $1=\nu_n(0,0) + \nu_n(0,s)$ if $s > 0$ and $\nu_n(0,0)$ if $s = 0$.
We have that $N(0,s) \leq N(0,0)$, since for every path $\gamma$ from $(0,s)$ to $(d,d)$, there is a path from $(0,0)$ to $(d,d)$ consisting in the vertical path from $(0,0)$ to $(0,s)$ and then $\gamma$.
Then, we have that 
\begin{align*}
	\rho_n
	& \leq N(0,0)\left(\nu_n(0,0) p_1^dp_2^d + \nu_n(0,s)p_1^dp_2^{d-s} \right) \\
	& \leq N(0,0) p_1^dp_2^{d-s}(\nu_n(0,s) + \nu_n(0,s)) = N(0,0) p_1^dp_2^{d-s}.
\end{align*}
Similarly, we have that 
$\rho_n \geq N(0,s) p_1^dp_2^d$.
On the other hand, the expected number of rounds needed to hit $(d,d)$ is equal to 
\[
\E(\calR) = \sum_{n=1}^{\infty} n \rho_n \prod_{k=1}^{n-1} (1-\rho_k),
\]
and therefore we can upper bound the value $\mathbb{E}(\calR)$ by
\begin{align*}
	\E(\calR)
	& \leq \sum_{n=1}^{\infty} n N(0,0)p_1^dp_2^{d-s} \left(1-N(0,s)p_1^dp_2^d\right)^{n-1} \\
	& = N(0,0)p_1^dp_2^{d-s}\left(\frac{1}{N(0,s)p_1^dp_2^d}\right)^2
	= \frac{N(0,0)}{N(0,s)^2p_1^dp_2^{d+s}}.
\end{align*}
Similarly, we can lower bound the value of $\mathbb{E}(\calR)$ to get
$\E(\calR)
\geq N(0,s) p_2^{2s}/(N(0,0)^2p_1^dp_2^d)$. 
\end{proof}
Note that when $s=0$ we get that the expected value of $\calR$ is exactly given by $1/N(0,0)p_1^dp_2^d$.
Furthermore, when $s < d/2$, the function $p_2^{2s}(p_1p_2)^{-d}$ attains its minimum at $p_1=d/(2(d-s))$ (subject to $p_1+p_2=1$ and $p_1,p_2\ge 0$). 
Therefore, using Proposition~\ref{prop:E(R)}, we get directly the following lower bound on $\E(\calR)$. 
\begin{corollary}\label{cor:LB}
	For $s < d/2$, we get that\;
	$\displaystyle \E(\calR)
	\geq \frac{N(0,s)}{N(0,0)^2}4^{d-s} \frac{(d-s)^{2(d-s)}}{d^d(d-2s)^{d-2s}}$.
\end{corollary}
Recall that the maximum gap tolerance is given by $\overline{g} = \max_{(\ell_1,\ell_2)\in \M} (\ell_2-\ell_1)$.
Note that the maximum is necessarily attained in $\partial_2 \M$, that is, at some state $(\ell_1,\ell_2)$ such that $C(\ell_1,\ell_2) = 1$.
We say that a capitulation policy $(C,s)$ has \emph{constant gap tolerance} $g \in \{1,\dots,d\}$, if $g(\ell) = g$ for all $\ell\in \{0,\ldots,d\}$.
We consider the following proposition borrowed from the integer lattice theory.
\begin{proposition}[{\cite[Theorem~10.3.3]{krattenthaler2015lattice}}]
	\label{prop:counting}
	Let $(C,s)$ be a capitulation policy with constant gap tolerance $g \in \{1,\dots,d\}$.
	Then, we have that
	\begin{align*}
		N(0,0) &= \sum_{k\in\Z} \left(\binom{2d}{d - k(g+2)} - \binom{2d}{d - k(g+2)+g+1}\right),\\
		N(0,s) &= \sum_{k\in\Z} \left(\binom{2d-s}{d - k(g+2)} - \binom{2d-s}{d - s - k(g+2)+g+1}\right).
	\end{align*}
\end{proposition}

Thanks to Proposition~\ref{prop:counting}, we obtain the following corollary.
\begin{corollary}
	For $d = 100$ and a tolerance of $T=10^8$ we have that constant gap tolerance strategies are safe for $0 \leq s \leq g \leq 4$, for any $p_1 \in (0,1)$.
\end{corollary}
In fact, for $d = 100$, there are additional cases, namely, $g=5$ and $s \leq 3$, or $(g,s) = (6,0)$, that also define safe strategies.
%
Moreover, it is possible to determine several safe strategies for $p_1$ in a given range as shown in Table~\ref{tab:safe-strategies}.
The examples shown are not extensive as for each one of the ranges for $p_1$ in Table~\ref{tab:safe-strategies}, it is possible to take higher values for $s$ (higher than $s_{\max}$ in Table~\ref{tab:safe-strategies}) for some smaller values of $g$ (smaller than $g_{\max}$ in Table~\ref{tab:safe-strategies}).
We can also exhibit examples of strategies which are unsafe, that is, capitulation policies such that the $\mathcal{T}(d,d)$ is small, namely, less than $T=10^4$ (about one month under $\tau_b = 5$ minutes).
Since $\mathcal{T}(d,d) \leq 2d \cdot \E(\calR)$, by Proposition~\ref{prop:E(R)}, we get 
\[
\mathcal{T}(d,d) \leq 2d \frac{N(0,0)}{N(0,s)^2p_1^dp_2^{d+s}}.
\]
Together with Proposition~\ref{prop:counting}, we can determine several unsafe strategies for a given $p_1$ as shown in Table~\ref{tab:unsafe-strategies}.
We remark that for $s\ge 2$, our methods are not accurate enough in order to detect possible unsafe strategies in that case.
\begin{table}[h]
	\centering
	\begin{tabular}{lrr}
		$p_1 \leq$  & $g_{\max}$ & $s_{\max}$ \\
		\hline\hline
		$0.45$ & $5$ & $3$ \\[1ex]
		$0.40$ & $5$ & $5$ \\[1ex]
		$0.\overline{33}$ & $d$ & $1$ \\[1ex]
		$0.30$ & $d$ & $5$ \\[1ex]
		$0.25$ & $d$ & $14$ \\[1ex]
		$0.20$ & $d$ & $30$ \\[1ex]
		$0.10$ & $d$ & $77$ \\[1ex]
		$0.05$ & $d$ & $d$ \\
	\end{tabular}
	\hspace{6em}
	\begin{tabular}{lrr}
		$p_1 \geq$ & $g_{\max}$ & $s_{\max}$ \\
		\hline\hline
		$0.65$ & $9$ & $1$ \\[1ex]
		$0.60$ & $6$ & $2$ \\[1ex]
		$0.\overline{66}$ & $d$ & $0$ \\[1ex]
		$0.70$ & $d$ & $2$ \\[1ex]
		$0.75$ & $d$ & $5$ \\[1ex]
		$0.80$ & $d$ & $9$ \\[1ex]
		$0.90$ & $d$ & $17$ \\[1ex]
		$0.95$ & $d$ & $23$ \\
	\end{tabular}
	\caption{Examples of safe strategies for $d=100$ and $T = 10^8$. For $p_1$ in a given range, $(s,g)$ is safe for every $g \leq g_{\max}$ and $s \leq \min\{g,s_{\max}\}$.}
	\label{tab:safe-strategies}
\end{table}
\begin{table}[h]
	\centering
	\begin{tabular}{c|c|*{11}{|c}}
		\multicolumn{2}{c||}{$p_1$} & $0.43$ & $0.44$ & $0.45$ & $0.46$ & $0.47$ & $0.48$--$0.52$ & $0.53$ & $0.54$ & $0.55$ & $0.56$ & $0.57$\\
		\hline\hline
		$s=0$ & $g_{\min}$ &
		$-$ &
		$19$ &
		$16$ &
		$14$ &
		$14$ &
		$13$ &
		$14$ &
		$14$ &
		$16$ &
		$19$ &
		$-$ \\
		\hline
		$s=1$ & $g_{\min}$ &
		$-$ &
		$-$ &
		$21$ &
		$17$ &
		$16$ &
		$15$ &
		$16$ &
		$19$ &
		$-$ &
		$-$ &
		$-$
	\end{tabular}
	\caption{Examples of unsafe strategies for $d=100$ and $T = 10^4$. For the given values of $p_1$ and $s$, $(s,g)$ is unsafe for every $g \geq g_{\min}$.}
	\label{tab:unsafe-strategies}
\end{table}
\section{A Market Share Case Study  \label{sec:Revenues}}
In what follows we analyse a family of capitulation regimes through the lens of the objective functions $R_i$ and $G_i$ defined in Section~\ref{subsec:PayoffFunctions}.
More precisely, we do this for the family of capitulation policies with constant gap tolerance $g\ge 1$ and $s=0$.
We further assume that $d \gg g$, which means that, in practice, we consider $d=\infty$.
By Proposition~\ref{prop:ConvergenceAverage}, we have that
\[R_i = \langle \pi,\hat{r}_i\rangle -  c_i\overline{\tau}\langle \pi_P, \hat{r}_1 + \hat{r}_2\rangle \; \text{ and } \; G_i = \frac{\langle \pi, \hat{r}_i\rangle}{\langle \pi_P, \hat{r}_1 + \hat{r}_2\rangle},\]
where $\pi = \pi_P$ is the corresponding stationary distribution.
Thus, our analysis can be reduced to the computation of $\rho_i = \langle \pi_P, \hat{r}_i\rangle$.
In particular, we have 
\begin{equation}
	R_i = \rho_i - c_i\overline{\tau}(\rho_1+\rho_2)\; \text{ and }\;G_i = \frac{\rho_i}{\rho_1 + \rho_2}\label{eq:GR}
\end{equation}
for each $i\in \{1,2\}$.
In order to compute explicitly the values of the objective functions $R_i$ and $G_i$, we compute the values $\rho_i$, that are completely determined by the corresponding stationary distributions.
To this end, we use a lattice path enumeration approach (see~\cite{krattenthaler2015lattice} for a survey).

In the stationary distribution, the value at any state of the Markov chain can be described as a linear combination of the value of $\pi_P$ at incoming states.
Thus, $\pi_P(\ell_1,\ell_2)$ is a linear combination of the number of interior lattice paths from $(0,0)$ and $(0,s)$ to $(\ell_1,\ell_2)$, weighted by the corresponding probabilities of such paths to occur.
Recall that in capitulation policies of constant gap tolerance $g \geq 1$ with fixed $s = 0$, player $1$ capitulates every time that player $2$ has mined $g$ blocks more than $1$.
Thus, the Markov chain has states $\M = \{(\ell,\ell),(\ell,\ell+k): \ell, k \in \NN, \, k \leq g\}$, and the capitulation states are 
$\partial_1 \M = \{(\ell,\ell): \ell \in \NN\}$ and $\partial_2 \M = \{(\ell,\ell+g): \ell \in \NN\}$.
See in Figure~\ref{fig:const-gap} the corresponding representation.
\begin{figure}[h!]
	\centering
	\includegraphics[width=.35\textwidth,height=.35\textheight]{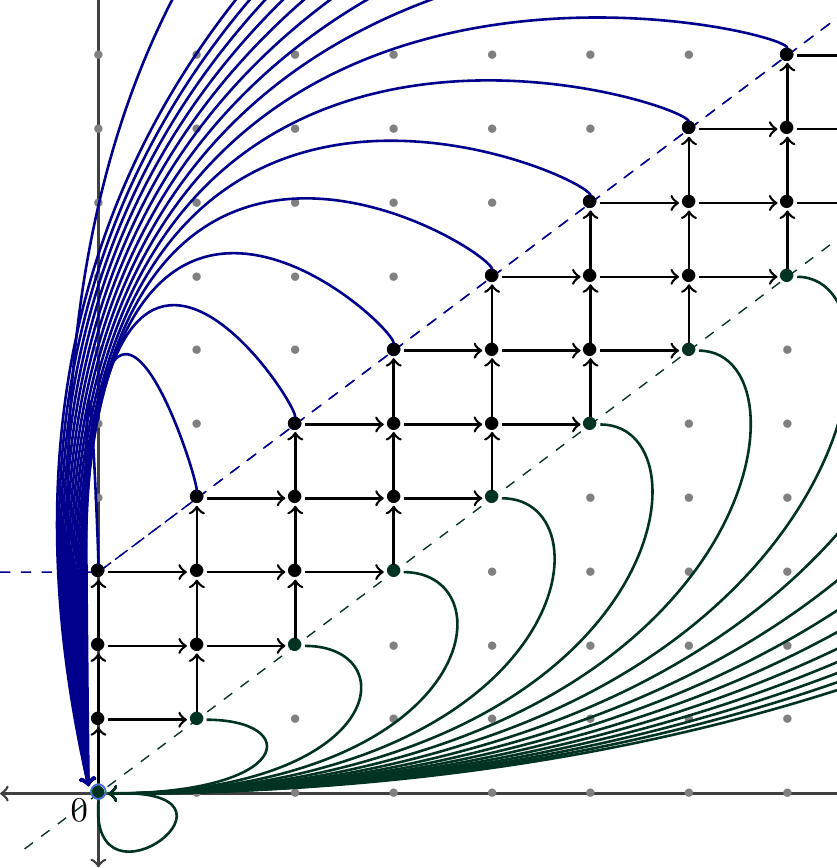}
	\caption{Lattice representation of a Markov chain of constant gap tolerance $g=3$, $s=0$ and $d=\infty$.}
	\label{fig:const-gap}
\end{figure}

When $d=\infty$, the states used to describe the game might be infinitely many. In this case, the capitulation policies might induce Markov chains $(\mathcal{M},P)$ with infinitely many states as well. While the chain $(\mathcal{M},P)$ keeps irreducible by construction, the ergodicity property described in Proposition \ref{prop:ConvergenceAverage} does not apply directly, since the existence of a stationary distribution is not guaranteed. However, the following proposition allow us to replicate the results of Proposition \ref{prop:ConvergenceAverage} if both players have bounded gap tolerance functions.

\begin{proposition}\label{prop:d-infty} 
Let $(C_1,s_1)$ and $(C_2,s_2)$ be two capitulation policies for a mining game with depth $d=\infty$. Suppose that both players have bounded gap tolerance, that is, there exists $\overline{g}$ such that for every $(\ell_1,\ell_2)\in \N\times\N$ the following holds:
When $\ell_2-\ell_1\geq \overline{g}$, we have  $C_2(\ell_1,\ell_2) = 1$,
and when $\ell_1-\ell_2\geq \overline{g}$, we have $C_1(\ell_1,\ell_2) = 1$.
Then, the induced Markov chain $(\mathcal{M},P)$ is positive recurrent, which yields that the conclusions of Lemma \ref{lemma:irreducible} and Proposition \ref{prop:ConvergenceAverage} hold.
\end{proposition}

\begin{proof} Since $(\mathcal{M},P)$ is irreducible, it is sufficient, according to \cite[Theorem 1.7.7]{Norris1998}, to show that at least one state $m\in \mathcal{M}$ is positive recurrent.
Let us consider the initial state $x=(s_2,0)$, and let $T$ be its first passage time, that is, $T = \inf\{ n\geq 1\ :\ X_n = (s_2,0) \}$.
Recall that $(s_2,0)$ is said to be positive recurrent if
$\E_{x}(T) <\infty$.
Define the random variables $(Z_n)_{n\in\N}$ over $(\Omega,\mathcal{F},\P_x)$ (see the definition of this probability space at Section \ref{subsec:Description}) as
\[
Z_n(m,\omega) = 1_{\{\omega_n = 1\}}(m,\omega),
\]
namely, $Z_n=1$ if player 1 won the mining race at turn $n$. Observe that, for any state $(\ell_1,\ell_2)\in\mathcal{M}$, in order to pass by $(s_2,0)$ it is enough for player 1 to have $2\overline{g}+1$ consecutive wins of the mining races. Let us define the random variables $(U_k)_{k\in\N}$ as
\[
U_k = \prod_{i=0}^{2g}Z_{(2\overline{g}+1)k+i}.
\]
We have that $(U_k)_{k\in \N}$ are independent, that $U_k$ follows a Bernoulli distribution of parameter $q = p_1^{2\overline{g}+1}$, and 
when $U_k(m,\omega) = 1$ we have that $T(m,\omega)\leq (2\overline{g}+1)(k+1)$.
Then, we can define the stopping time $T_{\text{seq}} = \inf\{ k\in \N :\ U_k = 1 \}$,
and in view of the previous implication, we can write
$T(m,\omega) \leq (2\overline{g}+1)(T_{\text{seq}}(m,\omega)+1)$.
Thus, by noting that $T_{\text{seq}}$ follows a geometric distribution of parameter $q$, we conclude that
\[
\E_x(T) \leq (2\overline{g}+1)(\E_x(T_{\text{seq}})+1)<\infty,
\]
and therefore $x=(s_2,0)$ is positive recurrent as we wanted to. Then, the conclusions follow by \cite[Theorem 1.7.7]{Norris1998} and \cite[Theorem 1.10.2]{Norris1998}.
\end{proof}
Let $L(\ell_1,\ell_2)$ be the number of interior lattice paths from $(0,0)$ to $(\ell_1,\ell_2)$.
We have
$\pi(\ell_1,\ell_2) =  \pi(0,0)L(\ell_1,\ell_2)p_1^{\ell_1} p_2^{\ell_2}$, where $\ell_2 \in \{\ell_1, \ell_1+1, \dots, \ell_1 + g\}$.
By \cite[Theorem~10.3.4]{krattenthaler2015lattice}, the number of lattice paths from $(0,0)$ to $(\ell,\ell+m)$ that are (weakly) above the line $\{(x,x):x\in \RR\}$ and below the line $\{(x,x+g):x\in \RR\}$, is given by
\begin{equation}
	L(\ell,\ell+m) = \frac{2}{g+2} \sum_{k=1}^{g+1} \sin\left(\frac{\pi k}{g+2}\right) \cdot \sin\left(\frac{\pi k(m+1)}{g+2}\right) \cdot \left(2\cos \left(\frac{\pi k}{g+2}\right)\right)^{2\ell+m}.\label{eq:paths}
\end{equation}
In the following proposition we summarize the exact values obtained for $\rho_1$ and $\rho_2$.
\begin{proposition}\label{prop:cap-policies}
Consider the capitulation policies with constant gap tolerance $g\ge 1$, $s=0$ and $d=\infty$, and let $\rho_i=\langle \pi_P,\hat r_i\rangle$ for each $i\in \{1,2\}$. 
Then, we have
\begin{align*}
\rho_1 &= {\tiny \frac{p_1}{\Gamma} \displaystyle\sum_{k=1}^{g+1} \frac{\sin^2(\frac{\pi k}{g+2})}{(1 - 4 p_1p_2 \cos^2 (\frac{\pi k}{g+2}))^2}},\text{ and}\\
\rho_2 &= \frac{p_2^{g+1}}{\Gamma} \displaystyle\sum_{k=1}^{g+1}(-1)^{k-1}{\tiny \frac{\sin^2(\frac{\pi k}{g+2}) (2\cos (\frac{\pi k}{g+2}))^{g}}{(1 - 4 p_1p_2\cos^2(\frac{\pi k}{g+2}))^2}\left(g\left(1 - 4p_1p_2\cos^2 \left(\frac{\pi k}{g+2}\right)\right)+1\right)},
\end{align*}
where $\Gamma=\sum_{k=1}^{g+1} \sin\left(\frac{\pi k}{g+2}\right) \sum_{m = 0}^{g}\sin\left(\frac{\pi k(m+1)}{g+2}\right)
\frac{ 2^m\cos ^m \left(\frac{\pi k}{g+2}\right)p_2^{m}}{1 - 4\cos^2 \left(\frac{\pi k}{g+2}\right) p_1 p_2 }$.
\end{proposition}
The proof of Proposition \ref{prop:cap-policies} can be found in the Appendix.
By February 2021, the value of 1 BTC was close to 47000 USD \footnote{\url{https://bitcoinmagazine.com/}} and the social energy consumption per year for Bitcoin mining has been estimated in 77.78 TWh per year.\footnote{\url{https://digiconomist.net/bitcoin-energy-consumption/}} Assuming the hash power of a player to be proportional to the energy consumption, and that the price of a kWh is 0.01 USD, we estimate the marginal costs of each $i\in \{1,2\}$ as
$c_i \approx 0.01/(6.25\cdot 4.7\cdot 10^4)\times (77.78\cdot 10^{9})/(365\cdot 24 \cdot 60)\times p_i \approx 0.005p_i \,\mathrm{[units/min]}$.  
When both players try to maximize their asymptotic expected revenue neglecting their costs, by Corollary~\ref{cor:EqWithoutCost}, Frontier is a Nash equilibrium. However, under Frontier regimes, we have that $\tau_b = \overline{\tau} = 10\,\text{min}$ and thus, the average cost of player $i$ becomes $c_i\tau_b \approx 0.05p_i$ units per turn, which is the 5\% of the average reward $p_i$. Thus, the question of profitability of selfish mining is definitely relevant and it has already been addressed by different authors (see e.g. \cite{profitability2020albrecher,profitability2020davidson}). 
The estimation of $c_i$ is rough, and we simply provide it to illustrate the ideas of this section. 
\begin{figure}[h!t]
	\centering
	\includegraphics[width=.45\textwidth]{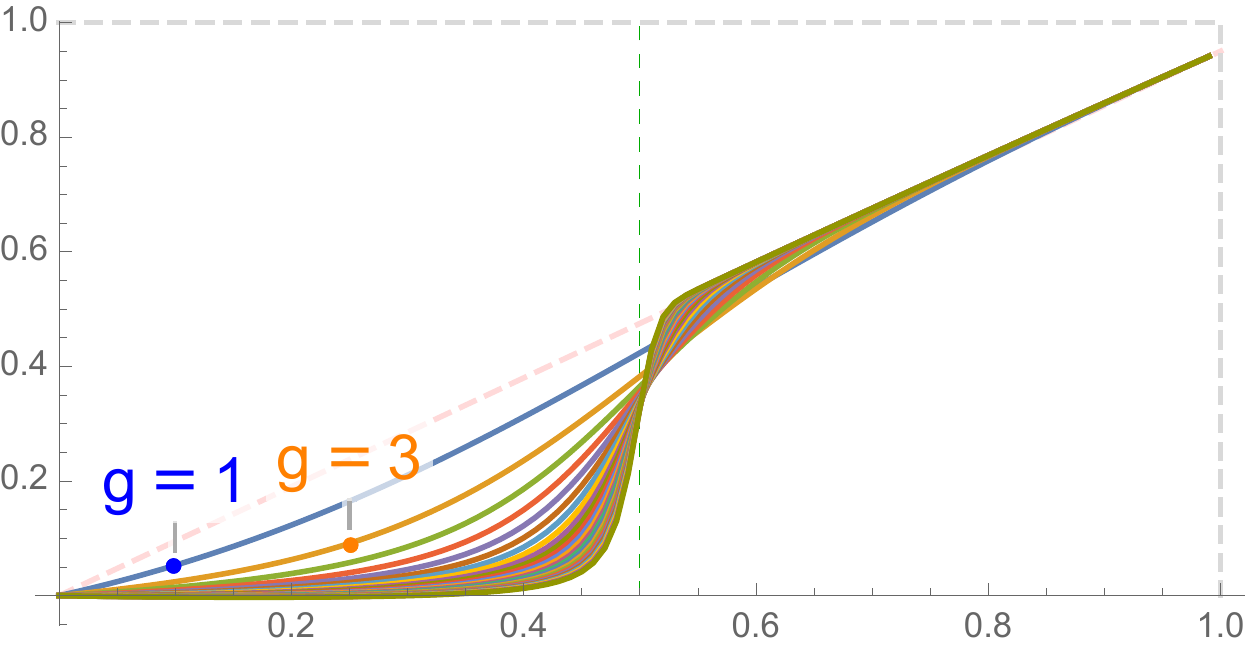}
	\hspace{.05\textwidth}
	\includegraphics[width=.45\textwidth]{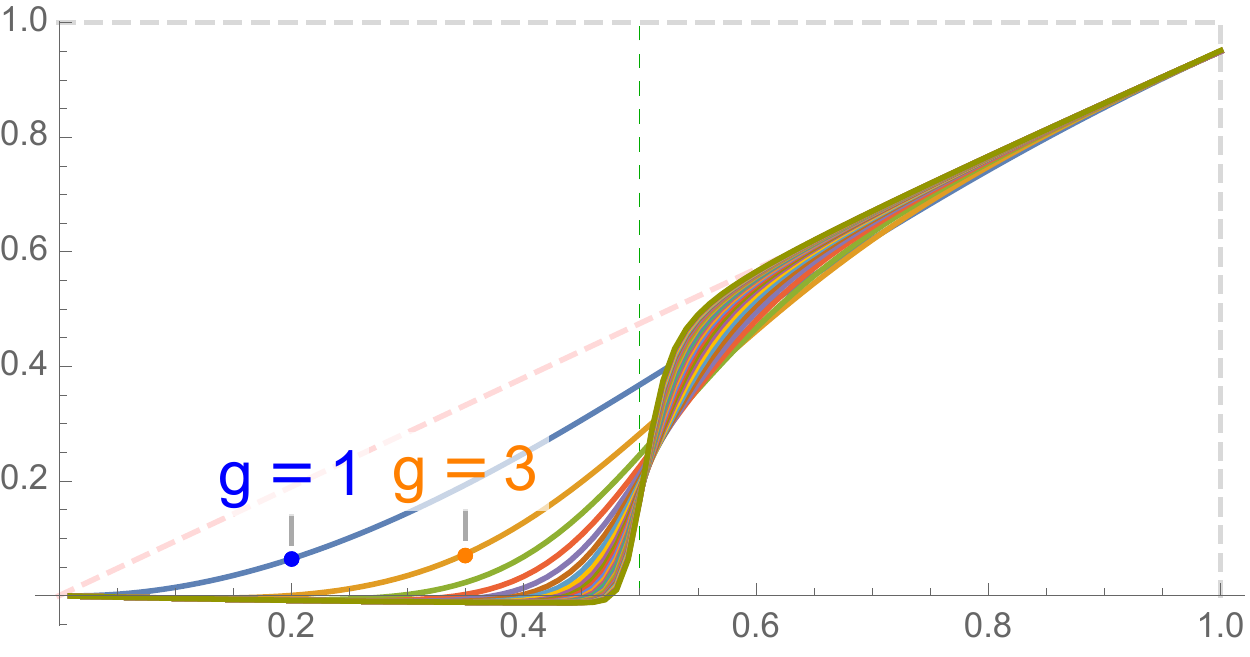}
	\caption{$R_1$ (left) and $R_2$ (right) as a function of $p_1$ and $p_2$, respectively, with $g \in \{1, 3, \dots, 49\}$ and $s=0$.
	}
	\label{fig:R-const}
\end{figure}

\begin{figure}[h!t]
	\centering
	\includegraphics[width=.45\textwidth]{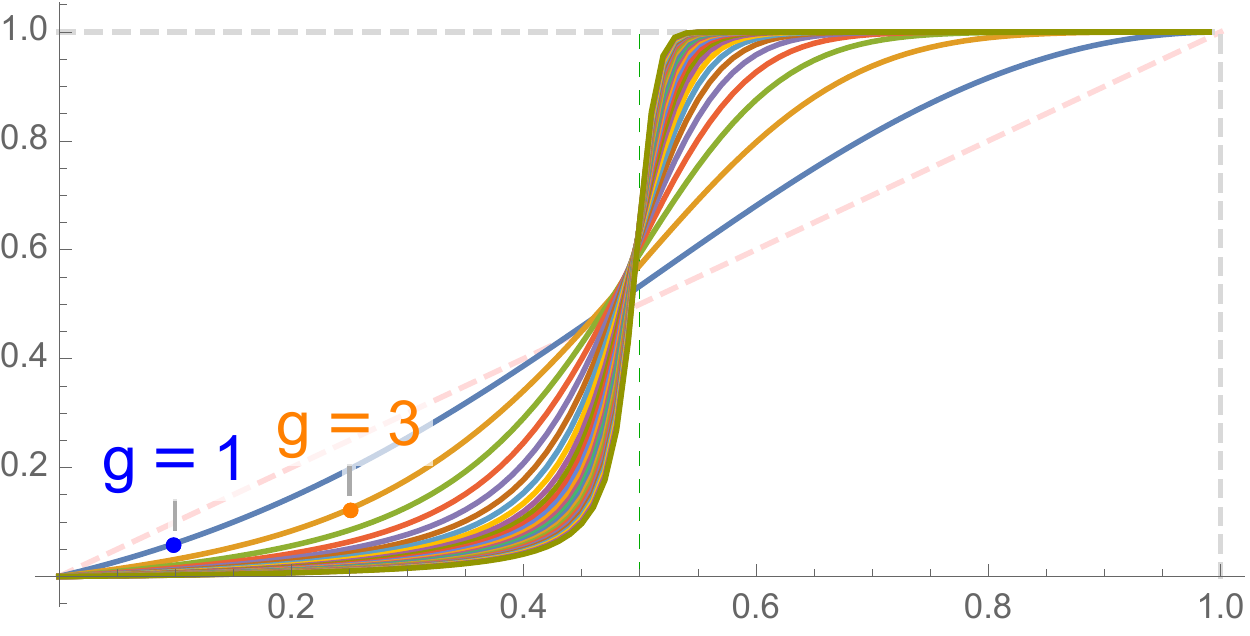}
	\hspace{.05\textwidth}
	\includegraphics[width=.45\textwidth]{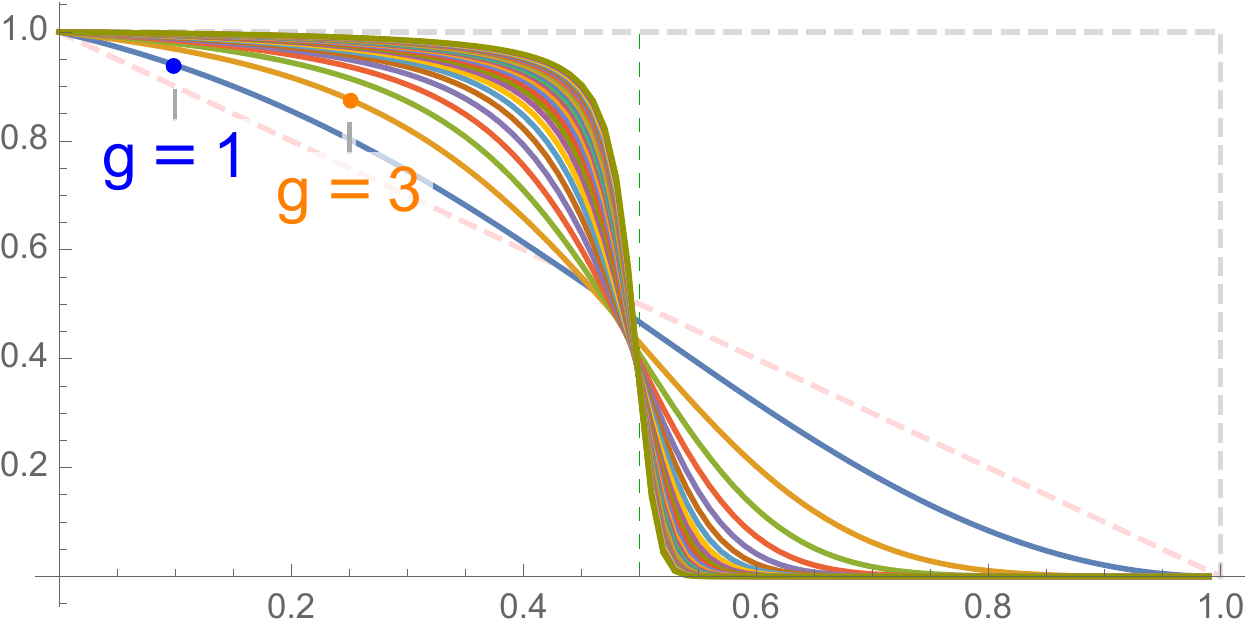}
	\caption{$G_1$ (left) and $G_2$ (right) as a function of $p_1$ and $p_2$, respectively, with $g \in \{1, 3, \dots, 49\}$ and $s=0$.}
	\label{fig:G-const}
\end{figure}

When everybody plays Frontier strategy, there is only one state, namely $(0,0)$, and we have $\rho_i = p_i$ for each player.
In this scenario, we get $R_i = 0.95p_i$ and $G_i= p_i$.
Observing how function $G_1$ evolves with respect to $p_1$ for different tolerances $g$, we see a similar behavior as in \cite{Koutsoupias2016}: around $p_1\approx 0.42$ strategic mining under immediate release becomes profitable (see Figure \ref{fig:G-const}). 
However, for $R_1$ a different behavior is observed: strategic mining is only profitable for values of $p_1$ larger than $0.5$ (se Figure \ref{fig:R-const}). Furthermore, very large tolerances (beyond the maximum gap described in Table \ref{tab:rapid-mixing}) are required for profitability when $p_1<0.6$. 
This suggests that in terms of profitability, strategic miners should decide using the average revenues per turn, rather than the ratio of validated blocks. 

Despite the discussion above, an interesting phenomenon occurs when we look at the revenues $R_2$ of the honest player (see Figure \ref{fig:RevHonest}). For sufficiently large tolerances, they become negative when $p_1$ is large enough. This means that a sufficiently powerful strategic miner with a very stubborn strategy (i.e. large gap tolerance) could eliminate honest ones from the market. For this to happen, rapid mixing is essential. We observe that if the strategic miner has $p_1>0.65$ and applies a capitulation policy of constant gap $g=10$,  then, within months (see Corollary \ref{coro:p-mix}), honest miners will lose money and quit the mining game in the long run. For even larger tolerances, the amount of power required to perform such strategy seems to approximate $0.5$. However, it is not clear if larger tolerances could be applied since fast mixing is not ensured. This situation deserves a further study since a monopolistic miner jeopardizes the very essence of decentralized cryptocurrencies. We aim to explore these perspectives in a subsequent work.

\begin{figure}
	\centering
	\includegraphics[width=0.45\textwidth]{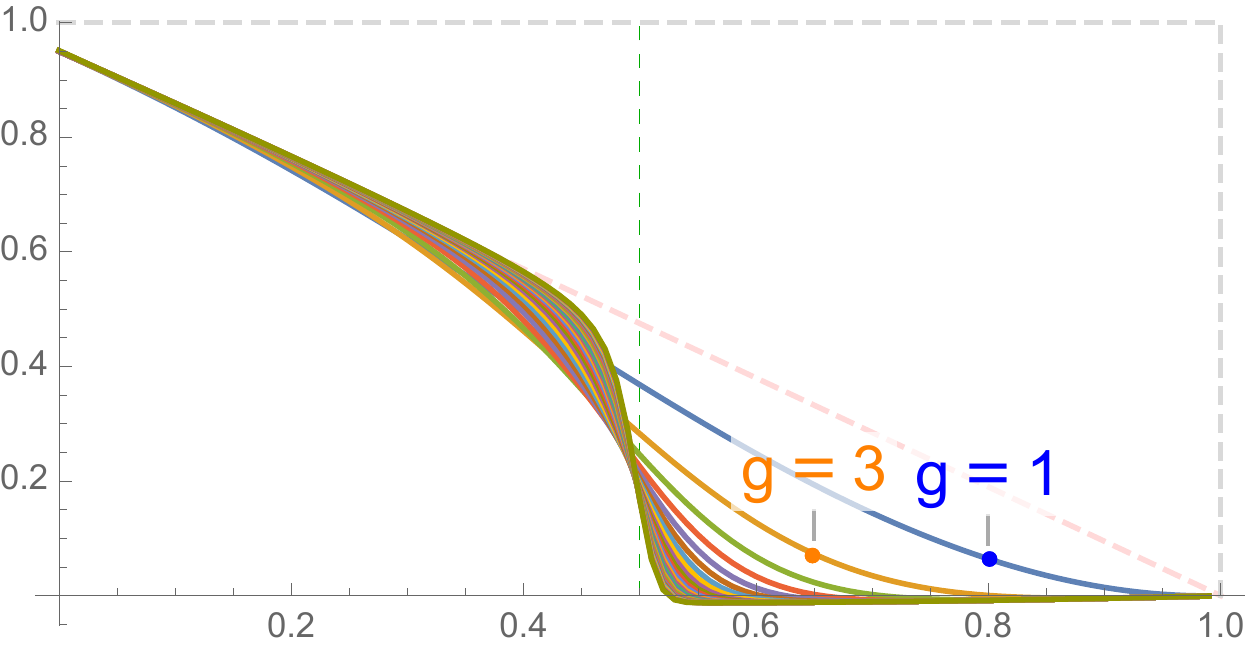}
	\caption{$R_2$ with respect to the power of the strategic miner ($p_1$) for different capitulation regimes of constant gap. For $p_1>0.5$, $R_2$ becomes negative for sufficiently large $g$.}
	\label{fig:RevHonest}
\end{figure}

\bibliographystyle{abbrv}
{\small \bibliography{bibMarkovModels}}


\appendix
\section{Appendix}

\begin{proof}[Proof of Proposition \ref{prop:cap-policies}] Since $\pi$ is a stationary distribution, and thanks to the path enumeration \eqref{eq:paths}, we have
		\begin{align*}
			1
			& = \pi(0,0) \sum_{\ell = 0}^{\infty}\sum_{m = 0}^{g} L(\ell,\ell+m) p_1^\ell p_2^{\ell+m} \\
			& = \pi(0,0) \sum_{\ell = 0}^{\infty}\sum_{m = 0}^{g}\frac{2}{g+2} \sum_{k=1}^{g+1} \sin\left(\frac{\pi k}{g+2}\right) \cdot \sin\left(\frac{\pi k(m+1)}{g+2}\right) \cdot \left(2\cos \left(\frac{\pi k}{g+2}\right)\right)^{2\ell+m}p_1^\ell p_2^{\ell+m} \\
			& = \pi(0,0) \frac{2}{g+2} \sum_{k=1}^{g+1} \sin\left(\frac{\pi k}{g+2}\right) \sum_{m = 0}^{g}\sin\left(\frac{\pi k(m+1)}{g+2}\right) \sum_{\ell = 0}^{\infty} \left(2\cos \left(\frac{\pi k}{g+2}\right)\right)^{2\ell+m}p_1^\ell p_2^{\ell+m} \\
			& = \pi(0,0) \frac{2}{g+2} \sum_{k=1}^{g+1} \sin\left(\frac{\pi k}{g+2}\right) \sum_{m = 0}^{g}\sin\left(\frac{\pi k(m+1)}{g+2}\right)
			\frac{ 2^m\cos ^m \left(\frac{\pi k}{g+2}\right)p_2^{m}}{1 - 4\cos^2 \left(\frac{\pi k}{g+2}\right) p_1 p_2 },
		\end{align*}
	%
	and therefore we get
	\[
	\pi(0,0) = \left(\frac{2}{g+2} \sum_{k=1}^{g+1} \sin\left(\frac{\pi k}{g+2}\right) \sum_{m = 0}^{g}\sin\left(\frac{\pi k(m+1)}{g+2}\right)
	\frac{ 2^m\cos ^m \left(\frac{\pi k}{g+2}\right)p_2^{m}}{1 - 4\cos^2 \left(\frac{\pi k}{g+2}\right) p_1 p_2 }
	\right)^{-1}
	\]
	and $\pi(\ell,\ell+m) =  \pi(0,0)L(\ell,\ell+m)p_1^{\ell} p_2^{\ell+m}$, for $m \in \{0, 1, \dots, g\}$.
	It follows that
	\[\rho_1 =p_1 \sum_{\ell=0}^{\infty} \pi(\ell,\ell) (\ell + 1)
	= p_1\pi(0,0) \sum_{\ell=0}^{\infty} L(\ell,\ell)p_1^{\ell} p_2^{\ell}(\ell + 1).\]
	We define
	$G_1(x)= \sum_{\ell=0}^\infty L(\ell,\ell)x^\ell (\ell+1)$, so that $\rho_1 = p_1 \pi(0,0) G_1(p_1p_2)$.
	Observe that $G_1(x) = xF_1'(x) +F_1(x)$, where
	\begin{align*}
		F_1(x)
		& = \sum_{\ell=0}^\infty L(\ell,\ell)x^\ell 
		= \sum_{\ell=0}^\infty \frac{2}{g+2} \sum_{k=1}^{g+1} \sin^2\left(\frac{\pi k}{g+2}\right) \cdot \left(2\cos \left(\frac{\pi k}{g+2}\right)\right)^{2\ell} x^\ell \\
		& =  \frac{2}{g+2} \sum_{k=1}^{g+1} \sin^2\left(\frac{\pi k}{g+2}\right) \sum_{\ell=0}^\infty \left(2\cos \left(\frac{\pi k}{g+2}\right)\right)^{2\ell} x^\ell 
		=  \frac{2}{g+2} \sum_{k=1}^{g+1} \frac{\sin^2\left(\frac{\pi k}{g+2}\right)}{1 - 4 x \cos^2 \left(\frac{\pi k}{g+2}\right)}.
	\end{align*}
	It follows that
	\[
	G_1(x) =  \frac{2}{g+2} \sum_{k=1}^{g+1} \frac{\sin^2\left(\frac{\pi k}{g+2}\right)}{\left(1 - 4 x \cos^2 \left(\frac{\pi k}{g+2}\right)\right)^2}
	\]
	and therefore we conclude that
	\[
	\rho_1 = \frac{p_1}{\Gamma} \displaystyle\sum_{k=1}^{g+1} \frac{\sin^2\left(\frac{\pi k}{g+2}\right)}{\left(1 - 4 p_1p_2 \cos^2 \left(\frac{\pi k}{g+2}\right)\right)^2}.
	\]
	Similarly, we have that
	\[\rho_2 =p_2 \sum_{\ell=0}^{\infty} \pi(\ell,\ell+g) (\ell + g + 1)
	= p_2\pi(0,0) \sum_{\ell=0}^{\infty} L(\ell,\ell+g)p_1^{\ell} p_2^{\ell+g} (\ell + g + 1)
	\]
	and consider
	$G_2(x)
	= \sum_{\ell=0}^\infty L(\ell,\ell+g)x^\ell (\ell+g+1)$
	so that $\rho_2 = p_2^{g+1} \pi(0,0) G_2(p_1p_2)$.
	We have that $G_2(x) = xF_2'(x)  + (g+1)F_2(x)$, where
	\begin{align*}
		F_2(x)
		& = \sum_{\ell=0}^\infty L(\ell,\ell+g)x^\ell 
		= \sum_{\ell=0}^\infty\frac{2}{g+2} \sum_{k=1}^{g+1} (-1)^{k-1}\sin^2\left(\frac{\pi k}{g+2}\right) \cdot \left(2\cos \left(\frac{\pi k}{g+2}\right)\right)^{2\ell+g} x^\ell \\
		& =  \frac{2}{g+2} \sum_{k=1}^{g+1}  (-1)^{k-1}\sin^2\left(\frac{\pi k}{g+2}\right) \left(2\cos \left(\frac{\pi k}{g+2}\right)\right)^{g} \sum_{\ell=0}^\infty \left(2\cos \left(\frac{\pi k}{g+2}\right)\right)^{2\ell} x^\ell \\
		& =  \frac{2}{g+2} \sum_{k=1}^{g+1}  (-1)^{k-1}\frac{\sin^2\left(\frac{\pi k}{g+2}\right) \left(2\cos \left(\frac{\pi k}{g+2}\right)\right)^{g}}{1 - 4 x\cos^2 \left(\frac{\pi k}{g+2}\right) }.
	\end{align*}
	Therefore, we have that $G_2(x)$ is equal to
	\begin{equation*}
		\frac{2}{g+2} \sum_{k=1}^{g+1}(-1)^{k-1}\frac{\sin^2\left(\frac{\pi k}{g+2}\right) \left(2\cos \left(\frac{\pi k}{g+2}\right)\right)^{g}}{\left(1 - 4 x\cos^2 \left(\frac{\pi k}{g+2}\right)\right)^2}
		\left(
		g\left(
		1 - 4x\cos^2 \left(\frac{\pi k}{g+2}\right)
		\right)
		+1
		\right),
	\end{equation*}
	from where we recover that
	\[
	\rho_2 = \frac{p_2^{g+1}}{\Gamma} \displaystyle\sum_{k=1}^{g+1}(-1)^{k-1}\frac{\sin^2\left(\frac{\pi k}{g+2}\right) \left(2\cos \left(\frac{\pi k}{g+2}\right)\right)^{g}}{\left(1 - 4 p_1p_2\cos^2 \left(\frac{\pi k}{g+2}\right)\right)^2}\left(g\left(1 - 4p_1p_2\cos^2 \left(\frac{\pi k}{g+2}\right)\right)+1\right).\qedhere
	\]
\end{proof}
\end{document}